\documentclass[letterpaper,conference]{ieeeconf}   

\usepackage{amsthm,amsmath,amssymb,mathtools,newtxmath}
\usepackage{graphicx}
\usepackage{cite}
\usepackage{cleveref}
\usepackage{algorithm,algpseudocode}

\newtheorem{assumption}{Assumption}
\newtheorem{proposition}{Proposition}
\newtheorem{lemma}{Lemma}
\newtheorem{theorem}{Theorem}
\usepackage{subcaption}
\newtheorem{remark}{Remark}
\newtheorem{definition}{Definition}

\usepackage{booktabs,lipsum}
\usepackage{color}
\usepackage{multirow}
\DeclareMathOperator*{\argmin}{arg\,min}

\DeclareMathOperator{\Tr}{Tr}

\IEEEoverridecommandlockouts                              

\title{\LARGE \bf
Reinforcement Learning for Linear Quadratic Control is Vulnerable Under Cost Manipulation
}

\author{Yunhan Huang$^{1}$ and Quanyan Zhu$^{1}$
\thanks{$^{1}$ Y. Huang and Q. Zhu are with the Department of Electrical and Computer Engineering,
        New York University, 370 Jay St., Brooklyn, NY.
         {\tt\small \{yh.huang, qz494\}@nyu.edu}}
}

\begin{document}

\maketitle
\thispagestyle{plain}
\pagestyle{plain}

\begin{abstract}
In this work, we study the deception of a Linear-Quadratic-Gaussian (LQG) agent by manipulating the
cost signals. We show that a small falsification of the cost parameters will only lead to a bounded change in the optimal policy. The bound is linear on the amount of falsification the attacker can apply to the cost parameters. We propose an attack model where the attacker aims to mislead the agent into learning a `nefarious' policy by intentionally falsifying the cost parameters. We formulate the attack's problem as a convex optimization problem and develop necessary and sufficient conditions to check the achievability of the attacker's goal.

We showcase the adversarial manipulation on two types of LQG learners: the batch RL learner and the other is the adaptive
dynamic programming (ADP) learner. Our results demonstrate that with only $2.296\%$ of falsification on the cost data, the attacker misleads the batch RL into learning the 'nefarious' policy that leads the vehicle to a dangerous position. The attacker can also gradually trick the ADP learner into learning the same `nefarious' policy by consistently feeding the learner a falsified cost signal that stays close to the actual cost signal. The paper aims to raise people's awareness of the security threats faced by RL-enabled control systems.
\end{abstract}

\section{Introduction}
The adoption of machine learning (ML), especially reinforcement learning (RL), in control theory and engineering enables the agent to learn a high-quality control policy without knowing the model or the cost criteria \cite{bradtke1994adaptive,mania2019certainty,lewis2009reinforcement,jiang2012computational,pang2019adaptive,gorges2019distributed}. The agent learns either by interacting with the environment through communication channels \cite{bradtke1994adaptive,lewis2009reinforcement} or by processing existing datasets from previous experience \cite{mania2019certainty}. 
However, the incorporation of learning into control enlarges the attack surface of the underlying system and creates opportunities for the adversaries. Adversarial parties can launch attacks such as Denial-of-Service attacks (DoS) \cite{cetinkaya2019overview,huang2021cross}, False Data Injection (FDI) attacks \cite{mo2010false,miao2016coding}, spoofing attacks \cite{zhang2017strategic,liu2020secure} on the communication channels, and data poisoning attacks \cite{ma2019policy,zhang2020online} on the existing dataset to mislead the agent and sabotage the underlying system. If not dealt with properly, such attacks can lead to a catastrophe in the underlying system. For example, self-driving platooning vehicles can collide with each other when measurement data is manipulated \cite{behzadan2019adversarial}. It is, hence, critical to study security threats to learning algorithms and design effective defense mechanisms to safeguard the learning-enabled control system.

The successful application of RL in control systems relies on accurate, timely, and consistent feedback from the environment or reliable datasets from experience. The accuracy, timeliness, and consistency of the feedback information are unlikely to be guaranteed in practice, especially in the presence of adversarial interventions. Without consistent or accurate feedback from the environment, the agent can either fail to learn an implementable control policy or be tricked into a `nefarious' control policy favored by the adversaries. In the last five years, there has been a surge in the number of research studies that focus on the security threats faced by RL with discrete state space \cite{chen2021adversarial, figura2021adversarial,huang2019deceptive,huang2021manipulating,ma2019policy,rakhsha2020policy,wang2020reinforcement,xu2021transferable,zhang2020adaptive}. Very few, if not none, of these studies have investigated the potential security threats that can impede RL-enabled control systems.

In this paper, we take the initiative to look into the security threats faced by RL-enabled control systems. In particular, we are interested in the deception of a Linear-Quadratic-Gaussian (LQG) agent by manipulating the cost signals. The attacker falsifies the cost signals received by the LQG agent to trick the agent into learning some `nefarious' policies favored by the attacker. In RL-enabled systems, the agent obtains cost signals from user feedback, as is illustrated in \ref{fig:DiagramCostMani}.
For example, in recommendation systems, the costs are often related to the number of user clicks, purchases, or add-to-cart rate \cite{wang2014exploration}; Cost signals are represented by user sentiment or engagement in RL-enabled dialogue generation \cite{li2016deep,zhang2020adaptive}; in human-robot interaction, cost signals are generated according to the level of dissatisfaction the human has while interacting with the robot \cite{modares2015optimized,priess2014solutions}. Such occasions open a back door for the attacker to influence the RL-enabled system by manipulating the cost signals. Beyond that, in networked control systems where the agent is remote to the plant, cost signals can be falsified or jammed by the attacker when transmitting from the plant to the agent \cite{huang2019deceptive,huang2021manipulating}. 

\begin{figure}[h]
\includegraphics[width=1.1\columnwidth]{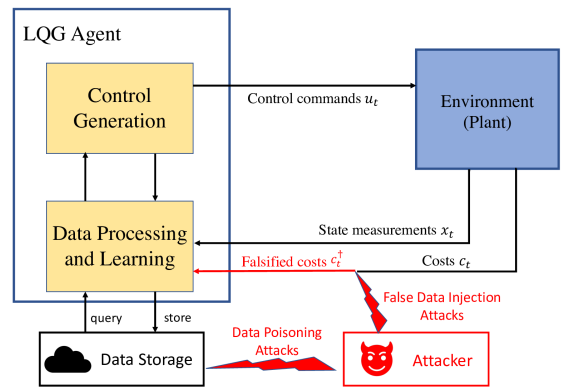}
\caption{Reinforcement learning for control under adversarial manipulation on cost signals.}
\label{fig:DiagramCostMani}
\end{figure}

We characterize the fundamental limits of cost manipulation, i.e., how much the attacker needs to falsify the cost signals to steer the learned control policy from an optimal one to a `nefarious' one favored by the attacker. Given a `nefarious' control policy the attacker aims for, we formulate the attack's problem as an optimization problem, which is proved to be convex. We show that the attacker cannot deceive the agent into learning some policies no matter how the attacker falsifies the cost signals. Hence, the optimization problem is infeasible for these policies. We develop a set of conditions in the frequency domain under which the attacker can mislead the agent into learning the `nefarious' policy.

We showcase the deception of two types of LQG learners: one is the batch RL learner \cite{mania2019certainty} and the other is the adaptive dynamic programming (ADP) learner \cite{bradtke1994adaptive}. The batch RL learner first estimates the system and cost parameters using a pre-collected dataset. The dataset includes the data points that record the state and control trajectories and corresponding cost signals. Then the agent computes the learned policy using the estimated system and cost parameters. Suppose the attacker can falsify the cost signals in the dataset. We show that by only falsifying $2\%$ of the cost signals (in magnitude), the attacker can trick the agent into a policy that steers a vehicle into a dangerous position.

Another is the ADP learner \cite{bradtke1994adaptive}. The ADP learner updates its estimates about the $Q$-function adaptively by interacting with the environment (plant) online. The ADP learner receives cost signals and state measurements from the environment and applies controls to excite the environment. Suppose the attacker can falsify the cost signals transmitted through some communication channels. We show that the attacker can craft attacks on cost signals by simply solving a convex optimization problem. The crafted attacks falsify the cost signals in a way such that the falsified cost signals and the true cost signals remain close when the system becomes stabilized. The experiment shows that the attacker can gradually mislead the agent into learning a `nefarious' policy favored by the attacker.

While most studies focus on improving the performance of RL algorithms in control (e.g., convergence rate, robustness, data efficiency, computational complexity), very few studies pay attention to the security threats faced by RL algorithms in control. The theories and examples presented in this paper demonstrate the vulnerabilities of RL-enabled control systems, which shows the necessity of investigating the potential security challenges faced by RL in control.

\textbf{Related works:}
There is a recent trend in studying security threats faced by RL algorithms \cite{chen2021adversarial, figura2021adversarial,huang2019deceptive,huang2021manipulating,ma2019policy,rakhsha2020policy,wang2020reinforcement,xu2021transferable,zhang2020adaptive,zhang2020robust,liu2021provably}.We can taxonomize these studies by the three types of attacks targeting at three different signals: attacks on the reward or cost signals\cite{huang2019deceptive,huang2021manipulating,ma2019policy,wang2020reinforcement,zhang2020adaptive}, attacks on the state sensing \cite{rakhsha2020policy,xu2021transferable,behzadan2019adversarial}, attacks on the action execution \cite{liu2021provably}. One can refer to Section 5 of \cite{huang2022reinforcement} for a brief review of this topic. Most of these studies focus on RL with discrete state and action spaces, and very few consider the security threats faced by RL-enabled control systems. Ma et al. \cite{ma2019policy} studies data poisoning attacks on the reward data on batch RL. Beyond the case of discrete state space, the authors also demonstrate the effectiveness of the attacks on an LQR learner using batch data. Our paper goes beyond batch RL and focuses on the deception of LQG agents through manipulating cost signals in a general setting. Hence, the results in this paper can be extended to different learning schemes and will not be limited to batch RL. In addition to demonstrating the effectiveness of the attacks through numerical experiments, we also develop theoretical underpinnings to understand the fundamental limits of such attacks, e.g., whether the attacker's goal can be achieved or not, how much falsification is needed to achieve such goals.

\textbf{Notation:}
Denote the set of non-negative real numbers by $\mathbb{R}_+$. Let $\mathbb{C}$ be the complex plane. 
Let $\mathbb{S}^n$ the set of all real symmetric matrices of order $n$. Denote the set of all positive semi-definite (respectively, positive definite) symmetric matrices by $\mathbb{S}_+$ (respectively, $\mathbb{S}_{++}$). Given $M,N\in\mathbb{S}$, $M\succeq N$ (respectively, $M\succ N$) means $M-N$ is positive semi-definite (respectively, positive definite). We denote $ I_n$  the identity matrix of order $n$, and $I$ is the identity matrix whose order depends on the context. Throughout the paper, prime denotes the transpose.

For a matrix $M\in\mathbb{S}$ of order $n$,  define $\Theta(M)$ as the half-vectorization of $M$:
$$
\Theta(M)\coloneqq \left[m_{1,1},\cdots,m_{1,n}, m_{2,2},\cdots, m_{2,n},\cdots,m_{n-1,n-1} \right]'.
$$
Here, $m_{i,j}$ is the element in the $i$-th row and the $j$-th column of $M$. For a vector $x\in\mathbb{R}^n$, define
$$
\bar{x} \coloneqq \left[x_1^2,2x_1x_2\cdots,2x_1x_n,x_2^2,\cdots,2x_{2}x_n,\cdots,x_n^2 \right]'.
$$

The Frobenius norm of a matrix is denoted by $\Vert \cdot \Vert_F$. The norm $\Vert \cdot \Vert$ refers to the Euclidean norm for vectors and the spectral norm for matrices unless specified otherwise. For a real matrix $M\in\mathbb{R}^{n\times n}$, the spectral radius of $M$ is denoted by $\rho(M)$. Define $\tau(M)\coloneqq \sup_{k\in\mathbb{N}} \{\Vert M^k \Vert/ \rho(M)^k\}$ as the smallest value such that $\Vert M^k\Vert \leq \tau(M,\rho(M))\rho(M)^k$ for all $k\geq 0$.

\textbf{Organization of the paper:} In Section \ref{Sec:LQGFormulation}, we introduce the LQG problem with general quadratic cost and present some preliminary results regarding the LQG problem. Section \ref{Sec:LQGManiCostPara} proposes the problem of adversarial manipulation on the cost parameters, the fundamental limits of what the attacker can or cannot achieve. Section \ref{Sec:LQGManiCostPara} lays a theoretical foundation for the attack models on two popular learning methods: the Batch RL method and the ADP method, which are  introduced in \ref{Sec:BarchRLLearner} and \ref{Sec:DecevADPLearner}. We demonstrate the results using numerical examples in Section \ref{Sec:NumericalStudy}.

\section{LQG with General Quadratic Cost: Preliminaries and Background}\label{Sec:LQGFormulation}

Consider the discrete-time, multi-variable system 
\begin{equation}\label{Eq:SystemDynamic}
x_{t+1} = g(x_t,u_t,w_t) \coloneqq A x_t  + Bu_t + Cw_t,
\end{equation}
where $x_t\in\mathbb{R}^n,t=0,1,\cdots$ is the system state, $u_t \in\mathbb{R}^m,t=0,1,\cdots$ is the control input, $w_t\in\mathbb{R}^q$ is drawn $i.i.d.$ from the standard Gaussian distribution $\mathcal{N}(0,\sigma^2 I_q)$, $A\in\mathbb{R}^{n\times n}$, $B\in\mathbb{R}^{n\times m}$, and $C\in\mathbb{R}^{n\times q}$.

Consider a stationary control policy $\pi:\mathbb{R}_n\rightarrow \mathbb{R}_m$ taking the form
\begin{equation}
u_t = \pi(x_t) = Kx_t +k,
\end{equation}
where $K\in\mathbb{R}^{m\times n}$ and $k\in\mathbb{R}^m$. The stage-wise cost of the system is quadratic
\begin{equation}\label{Eq:ImmediateCostFun}
c_t = c(x,u) =  x' D x + d'x  + r + u'Eu,
\end{equation}
for some positive semi-definite $D\in\mathbb{S}_+$, positive definite $E\in \mathbb{S}_{++}$, vector $d\in\mathbb{R}^n$, and scalar $r$.
\begin{assumption}\label{Ass:ConllabilityObservability}
Assume that $(A,B)$ is controllable and $(A,D^{1/2})$ is observable. 
\end{assumption}

\begin{assumption}
$B$ has full column rank.
\end{assumption}

\begin{definition}
We say $K\in\mathbb{R}^{m\times n}$ is stabilizing if the matrix $A+BK$ is Schur stable, i.e., $\rho(A+BK)<1$.
\end{definition}

We consider the total cost as the discounted accumulated costs over an infinite horizon. Starting at state $x_t$ under control policy characterized by $K$ and $k$, the total cost is $V_{K,k}(x_t) = \mathbb{E}[\sum_{i=0}^\infty \gamma^i c_{t+i}|x_t]$, where $0\leq \gamma \leq 1$ is the discount factor and the expectation is over $\{w_t,t=0,1,2,\cdots\}$.

It is well known that if $K$ is stabilizing, $V_{K,k}$ takes the form \cite{bertsekas1987dynamic}:
$$
V_{K,k}(x_t) = x_t' P_{K} x_t + h_{K,k}'x + l_{K,k},
$$
Note that the subscript of $P_K$ is $K$ instead of $K,k$ because $P_K$ only depends on $K$. Let $K^*$ and $k^*$ characterize the policy which is optimal in a sense that the total discount cost of every state is minimized. The value function is defined by
$$
V^*(x) \coloneqq V_{K^*,k^*}(x) = x' P^* x + {h^*}'x + l^*,
$$
where $P\in\mathbb{S}_{++}$, $h^*\in\mathbb{R}^n$, and $l^*\in\mathbb{R}$.
It is well known that the optimal policy can be characterized by \cite{bertsekas1987dynamic}
$$
u_t^* = K^*x_t + k^*,
$$
where 
\begin{equation}\label{Eq:ControllerParameters}
\begin{aligned}
K^* &= -\gamma (E+\gamma B'P^*B)^{-1} B'P^*A,\\
k^* &= -(\gamma/2) (E+\gamma B'P^*B)^{-1} B'h^*,
\end{aligned}
\end{equation}
where 
\begin{align}\label{Eq:RiccatiEquation}
P^* &= D + \gamma A'P^*A - \gamma^2 A'P^* B(E + \gamma B'P^*B)^{-1}B'P^*A,\\
h^* &= d + \gamma (A+BK^*)'h^*, \label{Eq:LinearEquation}\\ 
l^* &= \frac{ r + \frac{\gamma}{2} \Tr(\Sigma_wC'P^*C) - \frac{\gamma^2}{2}  {h^*}'B(E + \gamma B'P^*B)^{-1} B'{h^*}}{1-\gamma}. \nonumber
\end{align}

\begin{lemma}\label{Lemma:PolicyUniquelyDecided}
Under Assumption \ref{Ass:ConllabilityObservability}, the pair $(K^*,k^*)$ is uniquely decided by the system parameters $(A,B)$ and the cost parameters $(D,E,d)$. \end{lemma}
The proof of Lemma \ref{Lemma:PolicyUniquelyDecided} is presented in Appendix \ref{Proof:PolicyUniquelyDecided}. Since the optimal policy $(K^*,k^*)$ is uniquely determined by $(A,B)$ and $(D,E,d)$, for a given system $(A,B)$ satisfying Assumption \ref{Ass:ConllabilityObservability}, we can define the an auxiliary notation for the solution $(K,k)$ to the discrete-time LQG as 
$$
(K,k) = \mathrm{DLQG}(D,E,d,A,B),
$$
where the mapping $\mathrm{DLQG}:\mathbb{S}_{+}^n\times \mathbb{S}_{++}^m\times \mathbb{R}^n \times \mathbb{R}^{n\times n}\times \mathbb{R}^{n\times m}\rightarrow \mathbb{R}^{m\times n} \times \mathbb{R}^{m}$ is well-defined and characterized by the relations (\ref{Eq:ControllerParameters})-(\ref{Eq:LinearEquation}). Since we are particularly interested in the manipulation of the cost signals, we write $(K,k) = \mathrm{DLQG}(D,E,d)$ for simplicity. 

\subsection{Q-Functions and Policy Improvement}\label{Subsec:IntroPolicyIter}
Define the $Q$-function for a stabilizing policy $K,k$ as
\begin{equation}\label{Eq:DefineQFun}
Q_{K,k}(x,u) = c(x,u) + \gamma \mathbb{E}_w[V_{K,k}(g(x,u,w))|x]
\end{equation}
the value $Q_{K,k}(x,u)$ is the immediate cost of taking control $u$ from state $x$ plus the expected cost-to-go starting at the next state $Ax+Bu+Cw$. And $Q_{K,k}:\mathbb{R}^n\times \mathbb{R}^m\rightarrow \mathbb{R}$ is defined for all states $x$ and all admissible controls $u$. The expression can be written as
\begin{equation}\label{Eq:QFunQuadForm}
\begin{aligned}
&Q_{K,k}(x,u)\\
=& x' D x + d'x  + r + u'Eu + \gamma \mathbb{E}_w[ V_{K,k}(g(x,u,w))]\\
=& \begin{bmatrix} x' & u' & 1
\end{bmatrix} \begin{bmatrix}
 H_{K(xx)} & H_{K(xu)} & H_{K,k(x1)}\\
 H_{K(ux)} & H_{K(uu)} & H_{K,k(u1)}\\
 H_{K,k(1x)} & H_{K,k(1u)} & H_{K,k(11)}\\
 \end{bmatrix} \begin{bmatrix}
 x\\ u\\1
 \end{bmatrix},\\
=& \begin{bmatrix} x' & u' & 1
\end{bmatrix} H_{K,k} \begin{bmatrix}
 x\\ u\\1
 \end{bmatrix},
\end{aligned}
\end{equation}
where $H_{K,k}$ is a symmetric positive definite matrix.

Given the policy $K$, $k$, we can compute the $Q$-function $Q_{K,k}(\cdot,\cdot)$ characterized by $H_{K,k}$. We can find an improved policy $\tilde{K}$ and $\tilde{k}$ based on $Q_{k,k}(\cdot,\cdot)$, by solving
$$
\tilde{K}x + \tilde{k} = \min_{u} Q_{K,k}(x,u).
$$
Solving the minimization problem yields:
$$
\begin{aligned}
\tilde{K} &= - H_{K(uu)}^{-1} H_{K(ux)},\\
\tilde{k} &= - H_{K(uu)}^{-1} H_{K,k(u1)}.
\end{aligned}
$$
If the policy characterized by $K$ and $k$ is stabilizing, the feedback policy characterized by $\tilde{K}$ and $\tilde{k}$ is per definite a stabilizing policy --- it has no higher cost than $K$ and $k$ \cite{bertsekas1987dynamic,bradtke1994adaptive}. A new $Q$-function then be assigned to this improved policy and the policy improvement procedure can be repeated ad infinitum. Algorithms following this policy improvement is called policy iteration.

\section{LQG with Manipulated Cost Parameters}\label{Sec:LQGManiCostPara}

RL learning algorithms rely on data $\mathcal{D} = \{(x_t,u_t,c_t,x_{t+1}),t=1,2,\cdots,T\}$ to solve the LQG problem and find a good control policy with certain performance guarantee. For some model-based learning algorithms, the agent first estimates the system parameters $(A,B)$ and $(D,E,d)$ and then computes the `optimal' policy based on the estimates 
$(\hat{A},\hat{B})$ and $(\hat{D},\hat{E},\hat{d})$: 
$(\hat{K},\hat{k}) = \mathrm{DLQG}(\hat{D},\hat{E},\hat{d})$. For RL learning algorithms based on value iteration or policy iteration which will not conduct system identification, the attacker can rely on the mapping $\mathrm{DLQG}$ to craft its attacking strategy and understand the limitations of the attack. 

Hence, to understand the security threats of RL-based LQG problems, it is essential to investigate some fundamental properties regarding the mapping $\mathrm{DLQG}$:
\begin{enumerate}
    \item To trick the agent into the learning a `nefarious' control policy $(K^\dag, k^\dag)$, does there exist a trio $(D^\dag,E^\dag,d^\dag)$ such that $(K^\dag,k^\dag)$ is optimal under system $(A,B)$ and cost parameters $(D^\dag,E^\dag,d^\dag)$, i.e., $(K^\dag,k^\dag) = \mathrm{DLQG}(D^\dag,E^\dag,d^\dag)$?
    \item How much falsification the attacker needs to make on $(D,E,d)$ to trick the agent into learning $(K^\dag,k^\dag)$?
    \item Will a small change in $(D,E,d)$ cause a significant change in $(K,k)$ for $(K,k) = \mathrm{DLQG}(D,E,d)$? Is $\mathrm{DLQG}$ Lipschitz continuous?
\end{enumerate}
Since the goal of the attacker is to mislead the agent into learning a stabilizing policy $(K^\dag,k^\dag)$, we assume that $D^\dag\in\mathbb{S}_{++}^n$, $E^\dag\in\mathbb{S}_{++}^m$ which result into a stabilizing policy.
In the following subsection, we answer these questions. We show that the mapping $\mathrm{DLQG}$ is locally Lipchitz continuous and derives an upper bound regarding how much falsification on $(D,E,d)$ is needed to deceive the agent into learning $(K^\dag,k^\dag)$ instead of $(K^*,k^*)$.

\subsection{Fundamental Limits}\label{Subsec:FunLim}

Note that $K$ is uniquely determined by $D$ and $E$ given $(A,B)$ and solving for $K$ involves solving the Riccati equation (\ref{Eq:RiccatiEquation}). The following theorem presents a perturbation analysis on the Riccati equation (\ref{Eq:RiccatiEquation}) to see how small falsification on $(D,E)$ induces changes in the solution of the Riccati equation $P^*$. 

\begin{proposition}\label{Prop:RiccatiPerturbation}
Let $\rho(A_c^*)$ be the spectral radius of $A_c^*\coloneqq A+BK^*$ under the true cost parameters. Let $D^*$ and $E^*$ be the true cost parameters and $D^\dag$ and $E^\dag$ are the falsified cost parameters. Suppose $\Vert D^\dag - D^*\Vert\leq \epsilon$ and $\Vert E^\dag - E^*\Vert\leq \epsilon$. Denote by $P^*$ (respectively, $P^\dag$) the solution to the Riccati equation (\ref{Eq:RiccatiEquation}) under the true cost parameters (respectively, the falsified parameters). Then, we have
$$\Vert P^\dag - P^*\Vert \leq \Gamma_1 \Vert D^\dag - D^* \Vert + \Gamma_2\Vert E^\dag - E^*\Vert$$
as long as
$$
\begin{aligned}
\epsilon \leq& 4 \gamma^2 \frac{({1-\Vert {E^*}^{-1} \Vert})\tau(A_c^*)^2}{1-\gamma \rho(A_c^*)^2}\cdot\min\{\frac{1-\gamma \rho(A_c^*)^2}{\gamma^2\tau(A_c^*)\Vert A_c^*\Vert \Vert S^*\Vert},1\}  \\ 
& \cdot \Big( \Vert A\Vert^{-2}(\Vert P^*\Vert +1)^{-2} {\Vert {E^*}^{-2}\Vert^2 \Vert B\Vert^2  }\Big)
\end{aligned},
$$
where 
$$
\Gamma_1=4 \gamma^2 \frac{\tau(A_c^*)^2}{1-\gamma \rho(A_c^*)^2}
$$
and
$$
\Gamma_2 =4 \gamma^2 \frac{\tau(A_c^*)^2}{1-\gamma \rho(A_c^*)^2}\Vert A\Vert^2(\Vert P^*\Vert +1)^2 \frac{\Vert {E^*}^{-1}\Vert^2 \Vert B\Vert^2  }{1-\Vert {E^*}^{-1} \Vert}.
$$
\end{proposition}

\begin{proof}
Given parameters $(D,E)$, let $F(X,D,E)$ be the matrix expression
$$
F(X,D,E) = X  - \gamma A'XA +\gamma^2 A'XB(E+\gamma B'X^*B)^{-1}B'XA-D
$$
Applying binomial inverse theorem, we can write $F(X,D,E)$ as
$$
F(X,D,E)=X -\gamma A' X\left(I +\gamma BE^{-1}B'X\right)^{-1}A-D.
$$
Let $P^*$ be the solution to the Riccati equation (\ref{Eq:RiccatiEquation}) under the true cost parameters $(D^*,E^*)$ and $P^\dag$ be the solution to the same Riccati equation under the falsified cost parameters $(D^\dag, E^\dag)$. Hence, $F(P^*,D^*,E^*) = F(P^\dag,D^\dag,E^\dag) = 0$.

For simplicity, define $S^* \coloneqq B {E^*}^{-1} B'$, $S^{\dag} = B {E^{\dag}}^{-1}B'$, and $A_c^* = A + BK^*$. Here, $K^*$ is the optimal feedback control gain under the true cost parameters $(D^*,E^*)$. By inspection, for any $X$ such that $I+\gamma S^*(P^*+X)$ is invertible, we can write
\begin{equation}\label{Eq:IntermediateEquality}
\begin{aligned}
F(P^*+X,D^*,E^*) &= F(P^*+X,D^*,E^*) - F(P^*,D^*,E^*)\\
& = F_X(X) + \mathcal{H}(X),
\end{aligned}
\end{equation}
where 
$$
F_X(X)= X - \gamma{A_c^*}'XA_c^*,
$$
and 
\begin{equation}\label{Eq:IntermediateFun1}
\mathcal{H}(X) =\gamma^2 A_c^*X[I+\gamma S^*(P^*+X)]^{-1}S^*X A_c^*.
\end{equation}
Denote by $\Delta P$ the difference between $P^*$ and $P^\dag$, i.e., $\Delta P \coloneqq P^\dag - P^*$. Given $(D^*,E^*)$ and $(D^\dag,E^\dag)$, according to (\ref{Eq:IntermediateEquality}), the equation
\begin{equation}\label{Eq:IntermediateStep2}
F(P^*+X,D^*,E^*) -F(P^*+X,D^\dag,E^\dag) = F_X(X) + \mathcal{H}(X)
\end{equation}
admits a unique symmetric solution $X$ such that $P^* + X \geq 0$, which also solves $F(P^*+X,D^\dag,E^\dag)=0$. Hence, the solution is $X=\Delta P$. The eigenvalues of the operator $F_X:\mathbb{R}^{n\times n}\rightarrow \mathbb{R}^{n\times n}$ are $\mu_{ij} = 1-\lambda_i\lambda_j$, where the eigenvalues $\lambda_i$ of $A^*_c$ lies inside the unit circle in the complex plane. Hence $0 < |\mu_{ij}| <2$, the operator $F_X$ is invertible. In view of (\ref{Eq:IntermediateEquality}) and (\ref{Eq:IntermediateStep2}), we construct an operator
\begin{equation}\label{Eq:ConstructedMapping}
\Phi(Z) \coloneqq -F_X^{-1}(\mathcal{H}(Z) + F(P^*+Z,D^\dag,E^\dag)-F(P^*+Z,D^*,E^*)).
\end{equation}
Next, we show that under certain conditions on $\Delta D \coloneqq D^\dag - D^*$ and $\Delta E \coloneqq E^\dag - E^*$, there exists $\rho = f(\Vert \Delta D\Vert ,\Vert \Delta E\Vert )$ for some $f:\mathbb{R}\times \mathbb{R}\rightarrow \mathbb{R}$ such that $\Phi$ is constractive and maps the set
$$
\Omega_\rho = \{Z: \Vert Z\Vert \leq \rho,\ Z=Z', P+Z \geq 0\}
$$
into itself. In view of (\ref{Eq:ConstructedMapping}), we obtain
\begin{equation}\label{Eq:ConstructedMappingNorm}
\begin{aligned}
\Vert \Phi(Z)\Vert \leq  
\Vert F_X^{-1}\Vert\Big(& \Vert\mathcal{H}(Z)\Vert + \Vert F(P^*+Z,D^\dag,E^\dag)\\
&-F(P^*+Z,D^*,E^*)\Vert \Big).
\end{aligned}
\end{equation}
From Lemma \ref{Lemma:BoundFunctionNorm}, we know 
\begin{equation}\label{Eq:BoundPart1}
\Vert F_X^{-1} \Vert\leq  \frac{\tau(A^*_c)^2}{1-\gamma \rho(A^*_c)^2},
\end{equation}
where $\rho(A_c)$ is the spectral radius of $A_c$ and $\tau(A_c)$ is defined as $\tau(A_c)\coloneqq \sup_{k\in\mathbb{N}} \{\Vert A_c^k \Vert/ \rho(A_c)^k\}$. 

By Lemma \ref{Lemma:NormInquality1} and (\ref{Eq:IntermediateFun1}), 
\begin{equation}\label{Eq:BoundPart2}
\Vert \mathcal{H}(Z) \Vert \leq \gamma^2 \Vert A_c^*\Vert^2 \Vert S^*\Vert \Vert Z\Vert^2.
\end{equation}
A straight calculation yields
$$
\begin{aligned}
&F(P^*+Z,D^\dag,E^\dag) - F(P^*+Z, D^*,E^*)\\
=& -\gamma A'(P^\dag+Z)\left(I+\gamma S^\dag (P^\dag+Z)\right)^{-1}A \\
&+\gamma A'(P^\dag+Z)\left(I+\gamma S^* (P^\dag+Z)\right)^{-1}A-\Delta_D\\
=&\gamma \left[A'(P^\dag+Z) (I+
\gamma S^*(P^\dag+Z))\gamma \Delta S(I+\gamma S^\dag (P^*+Z))^{-1}A\right]\\ 
&- \Delta D. \\
\end{aligned}
$$
Then, by Lemma \ref{Lemma:NormInquality1},
\begin{equation}\label{Eq:BoundPart3}
\begin{aligned}
&\Vert F(P^\dag,D^\dag,E^\dag) - F(P^\dag, D^*,E^*) \Vert \\
\leq & \gamma^2 \Vert A\Vert^2 \Vert P^*+Z\Vert^2\Vert \Delta S\Vert + \Vert \Delta D\Vert.
\end{aligned}
\end{equation}
Assume $\rho<1$. Since $Z \in \Omega_\rho$, we obtain $\Vert P^*+Z\Vert \leq \Vert P^*\Vert+1$. Note that $\Delta S=S^\dag-S^*=B{E^\dag}^{-1}B' - B\{E^*\}^{-1}B'$. By Lemma \ref{Eq:NormInquality2}, we derive that if $\Vert {E^*}^{-1}\Delta E\Vert <1$, 
$$
\Vert \Delta S \Vert \leq \frac{\Vert {E^*}^{-1}\Vert^2 \Vert B\Vert^2  }{1-\Vert {E^*}^{-1} \Vert}   \Vert \Delta E\Vert.
$$
Combining the results from $(\ref{Eq:ConstructedMappingNorm})-(\ref{Eq:BoundPart3})$, we have, for $Z\in\Omega_\rho$
\begin{equation}\label{Eq:MapIntoSelfBound}
\begin{aligned}
\Vert \Phi(Z)\Vert \leq&  \frac{\tau(A^*_c)^2}{1-\gamma \rho(A^*_c)^2} \Bigg[ \gamma^2 \Vert A_c^* \Vert^2 \Vert S^*\Vert \rho^2 \\
&+ \gamma^2\Vert A\Vert^2 (\Vert P^*\Vert+1)^2 \frac{\Vert {E^*}^{-1}\Vert^2 \Vert B\Vert^2  }{1-\Vert {E^*}^{-1} \Vert}   \Vert \Delta E\Vert + \Vert \Delta D\Vert \Bigg]
\end{aligned}
\end{equation}
Similarly, we derive a bound for $\Vert \Phi(Z_1)-\Phi(Z_2) \Vert$ for $Z_1,Z_2\in\Omega_\rho$:
\begin{equation}\label{Eq:ContractionBound}
\begin{aligned}
\Vert \Phi(Z_1) - \Phi(Z_2) \Vert \leq& \frac{\tau(A^*_c)^2}{1-\gamma \rho(A^*_c)^2} 2 \gamma^2\Big[\Vert A \Vert^2 (\Vert P^* \Vert+1)^2 \Vert \Delta S \Vert \\
&+ (\Vert A_c^*\Vert^2 \Vert S^*\Vert)\rho
\Big] \Vert Z_1 - Z_2 \Vert
\end{aligned}
\end{equation}

Due to (\ref{Eq:MapIntoSelfBound}) and (\ref{Eq:ContractionBound}), the operator $\Phi$ is a contraction and maps the compact set $\Omega_\rho$ into itself if there exists $\rho>0$ such that 
$$
\begin{aligned}
\rho \geq &  \frac{\tau(A^*_c)^2}{1-\gamma \rho(A^*_c)^2} \Bigg[ \gamma^2 \Vert A_c^* \Vert^2 \Vert S^*\Vert \rho^2 \\
&+ \gamma^2\Vert A\Vert^2(\Vert P^*\Vert+1)^2 \frac{\Vert {E^*}^{-1}\Vert^2 \Vert B\Vert^2  }{1-\Vert {E^*}^{-1} \Vert}   \Vert \Delta E\Vert + \Vert \Delta D\Vert \Bigg],
\end{aligned}
$$
and 
$$
\begin{aligned}
1>&\frac{\tau(A^*_c)^2}{1-\gamma \rho(A^*_c)^2} 2
\gamma^2\Big[\Vert A \Vert^2 (\Vert P^* \Vert+1)^2 \frac{\Vert {E^*}^{-1}\Vert^2 \Vert B\Vert^2 }{1-\Vert {E^*}^{-1} \Vert}   \Vert \Delta E\Vert 
\\
&+  \Vert A_c^*\Vert^2 \Vert S^*\Vert \rho.
\Big]
\end{aligned}
$$
Choose 
$$
\begin{aligned}
\rho =& 4 \gamma^2 \frac{\tau(A_c^*)^2}{1-\gamma \rho(A_c^*)^2} \Big( \Vert A\Vert^2(\Vert P^*\Vert +1)^2 \frac{\Vert {E^*}^{-1}\Vert^2 \Vert B\Vert^2  }{1-\Vert {E^*}^{-1} \Vert}   \Vert \Delta E\Vert\\
&+ \Vert \Delta D\Vert \Big).
\end{aligned}
$$
If $\rho\leq \min\{\frac{1-\gamma \rho(A_c^*)^2}{\gamma^2\tau(A_c^*)\Vert A_c^*\Vert \Vert S^*\Vert},1\}$, the operator $\Phi$ is a contraction and maps the compact set $\Omega_\rho$ into itself. Then, $\Phi$ admits a unique fixed-point solution in $\Omega_\rho$. Therefore, $\Delta P \in \Omega_\rho$ and $\Vert \Delta P\Vert\leq \rho$. Hence, we have 
$$
\begin{aligned}
\Vert P^\dag -P^*\Vert \leq& 4 \gamma^2 \frac{\tau(A_c^*)^2}{1-\gamma \rho(A_c^*)^2} \Big(  \Vert \Delta D\Vert \\
&+ \Vert A\Vert^2(\Vert P^*\Vert +1)^2 \frac{\Vert {E^*}^{-1}\Vert^2 \Vert B\Vert^2  }{1-\Vert {E^*}^{-1} \Vert}   \Vert \Delta E\Vert \Big).
\end{aligned}
$$
\end{proof}

Proposition \ref{Prop:RiccatiPerturbation} presents some preliminary results showing the change in the solution of Riccati equation is bounded, i.e., $\Vert P^*- P^* \Vert \leq \Gamma_1 \Vert D^\dag - D^*\Vert + \Gamma_2 \Vert E^\dag -E^* \Vert$ for small $\Vert D^\dag - D^* \Vert$ and $\Vert E^\dag-E^* \Vert$. Our ultimate goal is to see how small fasification on $(D,E)$ leads to changes in $K$.

\begin{theorem}\label{Theo:BoundChangeK}
Suppose $\Vert E^\dag - E^*\Vert \leq \epsilon$ and $\Vert P^\dag - P^*\Vert\leq f(\epsilon)$, where $P^*$(respectively, $P^\dag$) is the solution to the Riccati equation (\ref{Eq:RiccatiEquation}) under the true cost parameters (respectively, the falsified cost parameters.), and $f:\mathbb{R}_+\rightarrow \mathbb{R}_+$ is some function on $\epsilon$ given $\Vert D^\dag - D^*\Vert \leq \epsilon$. If $\epsilon\leq \lambda_{\min}(E^*)/2$, then
\begin{equation}\label{Eq:BoundOnOptimalController}
\begin{aligned}
\Vert K^\dag-K^* \Vert \leq& \Gamma_3 f(\epsilon) + \Gamma_4 \epsilon,
\end{aligned}
\end{equation}
where $\Gamma_3 \coloneqq \frac{2\gamma}{\lambda_{\min}(E^*)}\max\{\Vert A\Vert,\Vert B\Vert \}^2(\Vert K^*\Vert +1)$ and $\Gamma_4\coloneqq \frac{2\gamma}{\lambda_{\min}(E^*)} \Vert K^*\Vert$.
\end{theorem}
\begin{proof}
Note that the optimal $Q$-function of the LQG problem takes the form 
$$
\begin{aligned}
Q(x,u)=&\frac{1}{2} x'Dx + d'x + r+\frac{1}{2} u' Eu\\
&+\gamma\mathbb{E}_w\left[ (Ax+Bu+Cw)P(Ax+Bu+Cw)\right],
\end{aligned}
$$
where $P$ the solution of the Riccati equation (\ref{Eq:RiccatiEquation}) given $A$, $B$, $D$, $E$. We know that the $\argmin_u Q(x,u)$ takes the form $Kx+k$ for every $x$. Indeed, $Kx$ is also the optimal point of the function $f(x,u)=\frac{1}{2}u'Eu+ \gamma (Ax+Bu)P(Ax+Bu)$. Applying Lemma \ref{Lemma:BoundonController} and letting $E_1= E^*$, $P_1=P^*$, $E_2= E^\dag$, and $P_2=P^\dag$, we have (\ref{Eq:BoundOnOptimalController}).
\end{proof}

From Proposition \ref{Prop:RiccatiPerturbation} and Theorem \ref{Theo:BoundChangeK}, we know that the change of the policy $K$ is bounded by the falsification in $(D,E)$. Indeed, 
$$
\Vert K^\dag -K^* \Vert \leq \Gamma_3\Gamma_1 \Vert D^\dag - D^* \Vert + (\Gamma_3\Gamma_2 +\Gamma_4)\Vert E^\dag -E^*\Vert,
$$
for small $\Vert D^\dag - D^*\Vert$ and $\Vert E^\dag - E^* \Vert$. Now, let's discuss the bound on $\Vert k^\dag - k^*\Vert$.

\begin{proposition}\label{Prop:PerturbationLinearh}
Let $h^*$ (respectively, $h^\dag$) be the solution of (\ref{Eq:ControllerParameters})-(\ref{Eq:LinearEquation}) under the true cost parameters $(D^*,E^*,d^*)$ (respectively, the falsified cost parameters $(D^\dag,E^\dag,d^\dag)$). Suppose $\Vert K^\dag - K^*\Vert \leq \frac{1}{2}\Vert B \Vert^{-1}\Vert I-\gamma A_c^*\Vert$. We have
$$
\Vert h^\dag - h^*\Vert\leq \Gamma_5 \Vert d^\dag -d^*\Vert  +\Gamma_6 \Vert K^\dag -K^*\Vert ,
$$
where $\Gamma_5 = 2 \Vert (I-\gamma A_c^*)^{-1} \Vert $ and $\Gamma_6 = 2\gamma \Vert (I-\gamma A_c^*)^{-1}\Vert \Vert d^*\Vert \Vert B\Vert$.
\end{proposition}
\begin{proof}
From (\ref{Eq:LinearEquation}), we know that
$$
h^* = d^* + \gamma (A+BK^*)'h^*\ \textrm{and}\ h^\dag = d^\dag + \gamma (A+BK^\dag)'h^\dag.
$$ 
Let $A_c^*= A+BK^*$ and $A_c^\dag = A+BK^\dag$. From Lemma \ref{Lemma:InvertibleMatrix}, we know that $I-\gamma A_c^*$ and $I-\gamma A_c^\dag$ are invertible. Hence, we have
$$
h^* = (I-\gamma A_c^*)^{-1} d^*\ \textrm{and }h^\dag = (I-\gamma A_c^\dag)^{-1} d^\dag
$$
Applying Lemma \ref{Lemma:PerturbationLinearSystem} proves the Theorem.
\end{proof}

\begin{theorem}\label{Theo:BoundOnk}
Let $k^*$ (respectively, $k^\dag$) be the solution of (\ref{Eq:ControllerParameters})-(\ref{Eq:LinearEquation}) under the true cost parameters $(D^*,E^*,d^*)$ (respectively, the falsified cost parameters $(D^\dag,E^\dag,d^\dag)$). Suppose $\Vert E^\dag - E^* \Vert\leq \lambda_{\min}(E^*)/2$, we have
$$
\Vert k^\dag - k^*\Vert \leq \Gamma_7 \Vert E^\dag - E^*\Vert + \Gamma_8 \Vert P^\dag - P^*\Vert + \Gamma_9 \Vert h^\dag -h^*\Vert ,
$$
where $\Gamma_7 = \frac{4}{\lambda_{\min}(E^*)} \Vert k^*\Vert $, $\Gamma_8 = \frac{4\gamma }{\lambda_{\min}(E^*)} \Vert k^*\Vert \Vert B\Vert^2$, and $\Gamma_9 = \frac{4\gamma }{\lambda_{\min}(E^*)}\Vert B\Vert$.
\end{theorem}
\begin{proof}
Note that 
From (\ref{Eq:ControllerParameters}), we know that $k = -\gamma (E+\gamma B'PB)^{-1}B'h$ is the solution of the following minimization problem
$$
\min_{k\in\mathbb{R}^m} f(k;E,P,h)\coloneqq k' (E+\gamma B'PB)k + \gamma h'Bk,
$$
where $E+\gamma B'PB$ is positive definite. Let $f^*(k)=f(k;E^*,P^*,h^*)$ and $f^\dag(k)=f(k;E^\dag,P^\dag,h^\dag)$. Let $k^*$ and $k^\dag$ be the solution of $\min_k f^*(k)$ and $\min_k f^\dag(k)$ respectively.
$$
\begin{aligned}
\Vert \nabla f^\dag(k) - \nabla f^*(k) \Vert \leq& 2\Vert k\Vert (\Vert E^\dag -E^*\Vert + \gamma \Vert B\Vert^2 \Vert P^\dag-P^*\Vert)\\
&+\gamma \Vert B\Vert \Vert h^\dag -h^*\Vert.
\end{aligned}
$$
From Lemma \ref{Lemma:SolutionBoundConvexFun}, we know that if $\Vert E^\dag - E^* \Vert\leq \lambda_{min}(E^*)/2$,
$$
\begin{aligned}
\Vert k^\dag - k^*\Vert \leq& \frac{2}{\lambda_{\min}(E^*)}\Vert \nabla f^\dag(k^*)\Vert \\
\leq& \frac{2 }{\lambda_{\min}(E^*)} \Big( 2\Vert k^*\Vert \Vert E^\dag - E^*\Vert\\
&+ 2\gamma \Vert k^*\Vert \Vert B\Vert^2 \Vert P^\dag - P^*\Vert+2\gamma \Vert B\Vert  \Vert h^\dag -h^*\Vert \Big).
\end{aligned}
$$
\end{proof}

Combining the results from Proposition \ref{Prop:RiccatiPerturbation}, Theorem \ref{Theo:BoundChangeK}, Proposition \ref{Prop:PerturbationLinearh}, and Theorem \ref{Theo:BoundOnk}, we obtain 
\begin{equation}\label{Eq:PolicyBoundedInAll}
\begin{aligned}
\Vert K^\dag -K^* \Vert &\leq \Gamma_3\Gamma_1 \Vert D^\dag - D^* \Vert + (\Gamma_3\Gamma_2 +\Gamma_4)\Vert E^\dag -E^*\Vert,\\
\Vert k^\dag -k^*\Vert &\leq \Gamma_1\Gamma_8 \Vert D^\dag - D^*\Vert +  (\Gamma_2\Gamma_8 +\Gamma_7)\Vert E^\dag -E^*\Vert \\&\ \ \ +\Gamma_5 \Gamma_9 \Vert d^\dag -d^*\Vert + \Gamma_6\Gamma_9\Vert K^\dag -K^*\Vert\\
&\leq (\Gamma_1\Gamma_8 + \Gamma_1\Gamma_3\Gamma_6\Gamma_9) \Vert D^\dag -D^*\Vert\\
&\ \ \ \ +\left[\Gamma_7 +\Gamma_2\Gamma_8 + ( \Gamma_4+ \Gamma_2 \Gamma_3 )\Gamma_6\Gamma_9\right] \Vert E^\dag -E^*\Vert\\
&\ \ \ \ +\Gamma_5\Gamma_9 \Vert d^\dag -d^*\Vert,
\end{aligned}
\end{equation}
for small $\Vert D^\dag - D^*\Vert$, $\Vert E^\dag - E^*\Vert$, and $\Vert d^\dag -d^*\Vert$ and for every $(D^*,E^*,d^*)\in\mathbb{S}_{+}^n\times \mathbb{S}_{++}^m\times \mathbb{R}^n$. The bound above indicates that the mapping $\mathrm{DLQG}$ that solves the discrete-time LQG problem is locally Lipschitz over the set $\mathbb{S}_{++}^n\times \mathbb{S}_++^m\times \mathbb{R}^n$. Hence, the continuity of the mapping $\mathrm{DLQG}$ holds and a tiny falsification on cost parameters by the attacker will only cause a bounded change in the computed policy. If the attacker aims to mislead the agent choosing into the `nefarious' policy $(K^\dag,k^\dag)$ from the original optimal policy $(K^*,k^*)$, (\ref{Eq:PolicyBoundedInAll}) gives an upper bound on how much falsification the attacker needs on the cost parameters.

\subsection{Attacks on Cost Function Parameters}

Suppose the attacker wants to deceive the agent into learning the policy characterized by $K^\dag$ and $k^\dag$ by altering the cost parameters. Results in Section \ref{Subsec:FunLim} give a rough bound on the region around the original cost parameters $(D,E,d)$ in which there might exist some $(D^\dag,E^\dag,d^\dag)$ that leads to the `nefarious' policy $(K^\dag, k^\dag)$. To minimize the cost of attacking, the attacker attempts to deviate the cost parameters from the original ones as small as possible:
\begin{equation}\label{Eq:CostParametersOptimization}
\begin{aligned}
\min_{\tilde{D},\tilde{d},\tilde{E},P,h}\ \ \ &\Vert \tilde{D} - D \Vert_{F} + \Vert \tilde{d} -d \Vert_{2} + \Vert \tilde{E}-E\Vert_{F},\\
s.t.\ \ \ &P = \tilde{D} + \gamma A'PA - {K^\dag}' (\tilde{E} + \gamma B'PB) K^\dag,\\
&(\tilde{E}+\gamma B'PB) K^\dag =-\gamma B'PA,\\
& h = \tilde{d} + \gamma (A+BK^\dag)'h,\\
&  2(\tilde{E}+\gamma B'PB) k^\dag = -\gamma B'h,\\
&P \succeq 0,\ \ \tilde{D}\succ 0,\ \tilde{E}\succ 0.
\end{aligned}
\end{equation}

\begin{proposition}
The optimization defined in (\ref{Eq:CostParametersOptimization}) is convex. 
\end{proposition}
\begin{proof}
It easy to see that the objective function of (\ref{Eq:CostParametersOptimization}) is convex. Suppose $(D_1,d_1,E_1,P_1,h_1)$ and $(D_2,d_2,E_2,P_2,h_2)$ satisfy the constraints in (\ref{Eq:CostParametersOptimization}). We need to show that any $0\leq \beta \leq 1$, $(\tilde{D},\tilde{d},\tilde{E},\tilde{P},\tilde{h})$ also satisfy the constraints in (\ref{Eq:CostParametersOptimization}), where $\tilde{D}= \beta D_1 + (1-\beta)D_2$, $\tilde{d}= \beta d_1 + (1-\beta)d_2$, $\tilde{E}= \beta E_1 + (1-\beta)E_2$, $\tilde{P}= \beta P_1 + (1-\beta)P_2$, and $\tilde{h}= \beta h_1$. We know that 
$$
\begin{aligned}
P_1 &= D_1 + \gamma A'P_1A - {K^\dag}' (E_1 + \gamma B'P_1B) K^\dag,\\
P_2 &= D_2 + \gamma A'P_2A - {K^\dag}' (E_2 + \gamma B'P_2B) K^\dag.\\
\end{aligned}
$$
Multiplying both sides of the first equality by $\beta$ and both sides of the second equality by $(1-\beta)$ yields
$$
\begin{aligned}
&\beta P_1 + (1-\beta) P_2 \\
=& \beta D_1 + (1-\beta)D_2 + \gamma A' [\beta P_1 + (1-\beta) P_2 ]A \\
&-{K^\dag}'[\beta E_1 + (1-\beta) E_2 + \gamma B'[\beta P_1 + (1-\beta) P_2]B ]K^\dag.
\end{aligned}
$$
This equality shows that
$$
\tilde{P} = \tilde{D} + \gamma A'\tilde{P}A - {K^\dag}' (\tilde{E} + \gamma B'\tilde{P} B) K^\dag.
$$
Similarly, we can show the rest of the constraints also form a convex set.
\end{proof}

Since the objective function of (\ref{Eq:CostParametersOptimization}) is strictly convex, there always exist a unique solution to the optimization problem (\ref{Eq:CostParametersOptimization}) if the constraints in (\ref{Eq:CostParametersOptimization}) are feasible for a given $(K^\dag,k^\dag)$.  The feasibility of problem (\ref{Eq:CostParametersOptimization}) is not always guaranteed. An interesting question is given $(A,B)$, under what conditions on $(K^\dag, k^\dag)$ such that the constraints in $(\ref{Eq:CostParametersOptimization})$ are feasible. This question is equivalent to asking given $(K^\dag,k^\dag)$ whether there exists $(\tilde{D},\tilde{E},\tilde{d})$ such that $(K^\dag,k^\dag) = \mathrm{DLQG}(\tilde{D},\tilde{E},\tilde{d})$. If $K^\dag$ is not stabilizing, $K^\dag$ cannot be optimal for any $D\succ 0$, $E\succ 0$, and $d^\dag$. Note that the goal of the attacker is not to unstablize the system but to mislead the agent into a nefarious policy $(K^\dag,k^\dag)$ with $K^\dag$ being stabilizing. But even if $K^\dag$ is stabilizing, there might not exist $D \succeq 0$ and $E\succ 0$ such that $K^\dag$ is optimal. Hence, our discussion will focus on under what conditions, a stabilizing $(K^\dag,k^\dag)$ is optimal for some $D\succeq 0$, $E\succ 0$, and $d$.

The same question was firstly raised by Kalman in the field of inverse optimal control \cite{kalman1964when}. He has showed
that for the single-input case, the circle criterion (see Theorem 6 of \cite{kalman1964when}) is a necessary and sufficient condition for a control policy to be optimal in a continuous-time system. Fujii and Narazaki have given a complete solution for the continuous-time multi-input case under the assumption that the cost parameters on the control input $E$ is fixed \cite{fujii1984complete}. Sugimoto and Yamamoto have showed the sufficient and necessary conditions for $K^\dag$ to be optimal for some $D$ and $E$ under the stage cost $x_t'Dx_t + u_t' Eu_t$ for a discrete-time system \cite{sugimoto1988solution}. We consider a more general cost function than \cite{sugimoto1988solution} did. Hence, we extend the results in \cite{sugimoto1988solution} to give the conditions on $(K^\dag, k^\dag)$ under which the constraints in (\ref{Eq:CostParametersOptimization}) are feasible.

In general, conditions of optimality can be expressed most conveniently using frequency-domain formulas \cite{kalman1964when}. For system $(A,B)$, consider the following right coprime factorization by two polynomial matrices $M(z)$ and $N(z)$:
$$
\sqrt{\gamma}(zI-A)^{-1}BE^{-1/2} = M(z)N(z)^{-1}.
$$
The feedback system $\sqrt{\gamma}(A+BK)$ induces a right coprime factorization
$$
\sqrt{\gamma}(zI-A_c)^{-1}BE^{-1/2}  = M(z)N_c(z)^{-1},
$$
where $N_c(z) = N(z)-KM(z)$. Define
$$
W(z)\coloneqq I_m + K(zI_n-A)^
{-1}BE^{-1/2}  = N_c(z)N(z)^{-1}.
$$
Define
$$
\Psi(z) \coloneqq N_c'(z^{-1})W(0)N_c(z) - N'(z^{-1})N(z).
$$
The following theorem states the conditions under which $(K,k)$ is optimal for some $D \succeq 0$, $E \succ 0$, and $d$.
\begin{theorem}\label{Theo: AttackFeasibilityCondition}
Let $(A,B)$ be controllable and suppose that $A$ is invertible and $B$ has rank $m$. Let the stabilizing control policy $u=Kx +k$ be given. Transform $\Psi(z)$ into the form 
$$
\Psi(z) U(z) =
\begin{bmatrix}
\tilde{\Psi}(z)\ & 0
\end{bmatrix}
$$
by some unimodular matrix $U(z)$ and, where $\tilde{\Psi}(z)$ is a rational function matrix with full rank. Then, the constraints in (\ref{Eq:CostParametersOptimization}) is feasible under the given control policy $(K,k)$ if and only if the conditions 
\begin{enumerate}
    \item $W(0)$ is positive definite,
    \item $\Psi(z)$ is positive semi-definite for all $z\in \mathbb{C}$ such that $|z|=1$,
    \item and there exists no pair of $\lambda\in \mathbb{U}^+$ and $v=(v_1,\cdots,v_m)'$ such that $N(\lambda)U(\lambda)v=0$ and $v_1=\cdots=v_m=0$,
    where $\mathbb{U}^+\coloneqq \{z\in\mathbb{C}:|z|\geq 1\}$,
\end{enumerate}
hold for some $E\succ 0$.
\end{theorem}
\begin{proof}
Note that whether $K$ is optimal depends only on $D$, $E$, $A$, and $B$. From (\ref{Eq:ControllerParameters}) and (\ref{Eq:RiccatiEquation}), we know that $K$ is optimal to $D$, $E$, $A$, and $B$ if 
$$
K= -\gamma (E+\gamma B'P^* B)^{-1}B'P^*A,
$$
where $P^*$ solves the Riccati equation $(\ref{Eq:RiccatiEquation})$. By Theorem 5.1 of \cite{sugimoto1988solution}, we know that given $(A,B)$ controllable and $A$ invertible, for a given stabilizing $K$, there always exist $D,P\succeq 0$ and $E \succ 0$ such that 
$$
\begin{aligned}
&P = D + \gamma A'PA - \gamma^2 A'P B(E + \gamma B'PB)^{-1}B'PA,\\
&(E+\gamma B'PB) K =-\gamma B'PA
\end{aligned}
$$
if and only if conditions 1), 2), and 3) are satisfied. For $k$, we have constraints
$$
\begin{aligned}
&k = -\gamma (E+\gamma B'PB)^{-1}B' h,\\
&h = (I-\gamma A_c)^{-1}d.
\end{aligned}
$$
Note that $B$ has full rank (rank $m$) and $E+\gamma B'PB$ and $I-\gamma A_c$ is invertible. Hence, for any $k$, as long as conditions 1), 2), 3) are satisfied for a given $K$, we can always find $d$ such that $k$ is optimal. We, hence, completes the proof.
\end{proof}

\begin{remark}
The assumption that $B$ has full column rank is reasonable. The assumption indicates that there is no redundant control inputs. As for the assumption that $A$ is invertible, consider a discrete-time system sampled from a continuous-time system $\dot{x} = \hat{A}x + \hat{B}u$ with a small sample period $\Delta t$. For the discretized linear system, we have $A= e^{\hat{A} \Delta t}$, $B= \int_{0}^{\Delta t} e^{\hat{A} \tau} \hat{B}d\tau$. Apparently, $A=e^{\hat{A} \Delta t}$ is invertible. 
\end{remark}

\begin{remark}
The conditions in Theorem \ref{Theo: AttackFeasibilityCondition} provide a quick way to check the feasibility of the optimization problem (\ref{Eq:CostParametersOptimization}) before solving it. If the optimization problem (\ref{Eq:CostParametersOptimization}) is not feasible, then the attacker cannot trick the agent into learning the `nefarious' policy $(K^\dag, k^\dag)$ no matter how the attacker falsifies the cost parameters. The feasible set of (\ref{Eq:CostParametersOptimization}) cannot be singleton. Once there exists $D\succeq 0$, $E\succ 0$, and $d$ such that the control policy $(K^\dag,k^\dag)$ is optimal, the same control policy is optimal for $\alpha D$, $\alpha E$, and $\alpha d$ for any $\alpha >0$.
\end{remark}

\begin{remark}
The conditions in Theorem \ref{Theo: AttackFeasibilityCondition} are stated in the frequency domain. Iracleous and Alexandridis have developed a set of necessary and sufficient conditions based on the state space representation in the time-domain, under which the control policy $u=Kx$ is optimal for some $D\succeq 0$, $E\succ 0$ under the linear system $(A,B)$ \cite{iracleous1998new} with stage cost $x_t'Dx_t + u_t' Eu_t$. One can extend their results to check the feasibility of (\ref{Eq:CostParametersOptimization}) using conditions in the time-domain.
\end{remark}



\section{Deceiving Batch RL Learner}\label{Sec:BarchRLLearner}

\subsection{The Batch RL Learner}

Consider a model-based LQG learner who learns a model from a training dataset batch first and then plans over the estimated model. The LQG learner implements system identification from a dataset and then computes the optimal controller based on the identified model \cite{mania2019certainty}. Consider the training dataset $\mathcal{D}=\{(x_t,u_t,c_t,x_{t+1}): t=0,1,\cdots,T-1\}$ generated from past experience of interacting with the LQR system. The learner implements system identification by solving the following least-square problems:
\begin{equation}\label{Eq:IdentifySystemPara}
(\hat{A},\hat{B}) \in \argmin_{A,B} \sum_{t=0}^{T-1} \frac{1}{2} \left\Vert  Ax_t + B u_t - x_{t+1} \right\Vert_2^2,
\end{equation}
\begin{equation}\label{Eq:CostParametersEstimation}
\begin{aligned}
(\hat{D},\hat{E},\hat{d},\hat{r})\in \argmin_{D\succeq 0, E\succ \epsilon I,d,r} \sum_{t=0}^{T-1}\Bigg\Vert&  x_t'Dx_t + d'x_t + r\\
&\ \ \ + u_t'Eu_t  -c_t  \Bigg\Vert.
\end{aligned}
\end{equation}
Define 
$$
X_T\coloneqq \begin{bmatrix}
x_1'\\
x_2'\\
\vdots\\
x_{T}'
\end{bmatrix},\ Z_T\coloneqq\begin{bmatrix}
z_0'\\
z_1'\\
\vdots\\
z_{T-1}'
\end{bmatrix},
$$
where $z_t\coloneqq \begin{bmatrix}
x_t' & u_t'
\end{bmatrix}'$. The least-square estimator for $(A,B)$ is (assuming the invertibility of $Z'Z$)
$$
\begin{bmatrix}
\hat{A} & \hat{B}
\end{bmatrix}'= (Z_T'Z_T)^{-1}Z_TX_T .
$$
Note that the optimization problem (\ref{Eq:CostParametersEstimation}) is convex. We can write $x_t'Dx_t$ as $\overline{x_t}' \Theta(D)$ and $u_t' Eu_t$ as $\overline{u_t} '\Theta(E)$. Define
$$
H \coloneqq \begin{bmatrix}
\overline{x_0}' & \overline{u_0}' & x_0' & 1\\
\overline{x_1}' & \overline{u_1}' & x_1' & 1\\
\vdots & \vdots & \vdots &\vdots  \\
\overline{x_{T-1}}' & \overline{u_{T-1}}' & x_{T-1}' & 1\\
\end{bmatrix}.
$$
The optimization problem $(\ref{Eq:CostParametersEstimation})$ admits a unique solution if the objective function is strictly convex, i.e., $H'H$ is positive definite, or equivalently $H$ is column independent. With the model identified, the learner computes the control policy based on (\ref{Eq:ControllerParameters}), (\ref{Eq:RiccatiEquation}), and (\ref{Eq:LinearEquation}) using the estimated system parameters $(\hat{A}, \hat{B})$ and $(\hat{D},\hat{E},\hat{d},\hat{r})$.

\subsection{The Batch RL Attacker}

We define the batch RL attack model by specifying the attacker's capability, objective, and the information he has:  
\begin{enumerate}
    \item The attacker is able to falsify the costs in the training dataset $\mathcal{D}$. The falsified training dataset is denoted by $\mathcal{D}^\dag =\{(x_t,u_t,c_t^\dag,x_{t+1}),t=0,1,\cdots,T-1\}$. Suppose the falsification of the cost data is consistent and compatible with the state and the control data, i.e., $c_t^\dag =  x_t'D^\dag x_t + {d^\dag}'x_t + r + u_t'E^\dag u_t$.
    \item The objective of the attacker is to trick the learner into learning the control policy $(K^\dag,k^\dag)$ by feeding the learner the falsified training dataset $\mathcal{D}^\dag$.
    \item  The attacker only has the knowledge of the original training dataset $\mathcal{D}$.
\end{enumerate}
Define $\mathbf{c} =[c_0,c_1,\cdots,c_{T-1}]'\in\mathbb{R}^T$ as a vector of cost signals from the original dataset $\mathcal{D}$. Let $\mathbf{c}^\dag =[c^\dag_0,c^\dag_1,\cdots,c^\dag_{T-1}]'\in\mathbb{R}^T$ be a vector of falsified cost signals in the poisoned dataset $\mathcal{D}^\dag$. To achieve his objective, the attacker needs to solve
\begin{align}
\min_{\tilde{D},\tilde{E},\tilde{d},\tilde{r},P,h} &\sum_{t=0}^{T-1} \Vert \mathbf{c}^\dag - \mathbf{c}\Vert_2 \label{Eq:ObjectiveBatchRL}\\ \label{s}
s.t.\ \ \ &P = \tilde{D} + \gamma \hat{A}'P\hat{A} - {K^\dag}' (\tilde{E} + \gamma \hat{B}'P\hat{B}) K^\dag,\\
&(\tilde{E}+\gamma \hat{B}'P\hat{B}) K^\dag =-\gamma \hat{B}'P\hat{A},\\
& h = \tilde{d} + \gamma (\hat{A}+\hat{B}{K}^\dag)'h,\\
& 2(\tilde{E}+\gamma \hat{B}'P\hat{B}) k^\dag = -\gamma \hat{B}'h,\\
&P \geq 0,\ \ \tilde{D}\geq 0,\ \tilde{E}> 0. \label{Eq:BatchRLPSD}\\
&c_t^\dag = x_t'\tilde{D} x_t + {\tilde{d}}'x_t + r + u_t'\tilde{E} u_t, \forall t
\end{align}

The attacker's problem defined in (\ref{Eq:ObjectiveBatchRL})-(\ref{Eq:BatchRLPSD}) is convex, whose feasibility is aligned with the feasibility of problem \ref{Eq:CostParametersOptimization} with $(A,B)$ replaced by $(\hat{A},\hat{B})$. The solution of the problem can then be used to falsify the original dataset $\mathcal{D}=\{(x_t,u_t,c_t,x_{t+1}): t=0,1,\cdots,T-1\}$ into the `bad' dataset $\mathcal{D}^\dag=\{(x_t,u_t,c^\dag_t,x_{t+1}): t=0,1,\cdots,T-1\}$. In Section \ref{Sec:NumericalStudy}, we will demonstrate the effectiveness of the batch RL attack model and the vulnerabilities of batch RL-enabled LQG systems.

\begin{remark}
The attacker uses $\hat{A}$ and $\hat{B}$ because he only knows the original dataset $\mathcal{D}$. Even if he knows $A$ and $B$, it is better to use $(\hat{A},\hat{B})$ to mimic the learner's problem.
Since $c_t^\dag =  x_t'\tilde{D} x_t + {\tilde{d}}'x_t + \tilde{r} +  u_t'\tilde{E} u_t$
$\mathbf{c}^\dag$ is completely decided by $\Tilde{D}$, $\tilde{E}$, $\tilde{d}$, $\tilde{r}$ and $\mathcal{D}$. Indeed, the attacker only needs to optimize over $\tilde{D}$, $\tilde{E}$, $\tilde{d}$, and $\tilde{r}$ even though we include $\mathbf{c}^\dag$ in the objective function for simplicity. \cite{ma2019policy} also considered data poisoning attacks on cost data. However, our formulation of the optimization problem differs from theirs in two aspects: 1. We require the cost data to be consistent and compatible with the state and the control to make the falsification less likely to be detected. 2. Our problem is convex. Hence, we can solve the problem to obtain the optimal solution instead of solving a surrogate problem which produces a sub-optimal value.
\end{remark}

\section{Deceiving ADP Learner}\label{Sec:DecevADPLearner}

\subsection{The ADP Learner}

The policy iteration algorithms presented in Section \ref{Subsec:IntroPolicyIter} will converge for the LQG problem \cite{kleinman1968iterative,bertsekas1987dynamic}. However, the  policy iteration algorithms required exact knowledge of the system model (\ref{Eq:SystemDynamic}) and the stage cost function (\ref{Eq:ImmediateCostFun}). Bradtke et al. proposed an ADP algorithm that allows the agent to perform an approximate version of the policy iteration algorithm using merely the observed data points $(x_t,u_t,c_t,x_{t+1}),t=0,1,2,\cdots$ \cite{bradtke1994adaptive}.

The ADP learner leverages Recursive Least Square (RLS) to directly estimate the function $Q_{K,k}$ in (\ref{Eq:QFunQuadForm}). To see how the adaptive learner learns the optimal policy, we rearrange (\ref{Eq:DefineQFun}) and (\ref{Eq:QFunQuadForm}) to obtain
$$
\begin{aligned}
c(x_t,u_t) &= Q_{K,k}(x_t,u_t) - \gamma Q_{K,k}(x_{t+1},u_{t+1})\\
&=\overline{[x_t,u_t,1]} \Theta(H_{K,k}) - \gamma \overline{[x_{t+1},u_{t+1},1]} \Theta(H_{K,k}),\\
&= \phi_t'\theta_{K,k},
\end{aligned}
$$
where $\phi_t=\overline{[x_t,u_t,1]} -\gamma \overline{[x_{t+1}, Kx_{t+1}+k,1]} \in \mathbb{R}^{(n+m+1)(n+m+2)/2}$, and $\theta_{K,k} = \Theta(H_{K,k}) \in\mathbb{R}^{(n+m+1)(n+m+2)2}$, where $[x,u,1]$ is the column vector concatenation of $x$, $u$, and scalar $1$. Through the lens of modern RL, $\overline{[x,u,1]}$ can be viewed as a feature vector of the original state and action space and  $\theta_{K,k}$ can be viewed as a vector of weights that the agent needs to tune to approximate the true $Q$-function.

RLS can now be used to estimate $\theta_{K,k}$. The recurrence relations for RLS are given by
\begin{equation}\label{Eq:RLS}
\begin{aligned}
\hat{\theta}_z(i) &= \hat{\theta}_z(i-1) + \frac{S_z(i-1)\phi_t \left( c_t - \phi_t'\hat{\theta}_z(i-1) \right)}{1+\phi_t'S_z},\\
S_z(i) &= S_z(i-1) - \frac{S_z(i-1) \phi_t \phi_t' S_z(i-1)}{1+\phi_t'S_z(i-1)\phi_t},\ \ S_z(0) = S_0.
\end{aligned}
\end{equation}
Here, $P_0= \beta I$ for some large positive constant $\beta$. $\theta_z = \Theta(H_{K_z,k_z})$ is the true parameter vector for the function $Q_{K_z,k_z}$. $\hat{\theta}_z(i)$ is the $i$th estimate of $\theta_z$. The subscript $t$ and the index $i$ are both incremented at each time step. And $z$ is the index that counts the number of policy updates the algorithm made. 

The LQG agent follow the $Q$-function based ADP algorithms (i.e., Algorithm \ref{Algo:ADPPolicyIter}) to find the optimal policy. In line 7, the agent adds an appropriate probing noise $e_t$ to make sure $\phi_t$ is persistently excited over time \cite{bradtke1994adaptive}. In line 9, when computing $\phi_t$, we let $u_{t+1} = K_{z} x_t + k_z$. In line 13 and 14, $\hat{H}_{(uu)}$, $\hat{H}_{(ux)}$, and $\hat{H}_{(u1)}$ are sub-matrices of the matrix $\hat{H}$ following the same notation for the matrix $H_{K,k}$ in (\ref{Eq:QFunQuadForm}). The convergence of Algorithm \ref{Algo:ADPPolicyIter} is guaranteed if $(A,B)$ is controllable (or at least stabilizable), $K_0$ is stabilizing, and $\phi_t$ is persistently excited. One can refer to \cite{bradtke1994adaptive} for more details about the $Q$-function based ADP algorithm.

\begin{algorithm}[H]
\caption{The $Q$-Function Based ADP Algorithm}{\label{Algo:ADPPolicyIter}}
\begin{algorithmic}[1]
\State \textbf{Initialize:} Stabilizing policy $(K_0,k_0)$, tolerance $\epsilon_1>0$, $\epsilon_{2}>0$
\State Set $\hat{\theta}_0(0)=\mathbf{0}$, $z=0$, and $t=0$
\Repeat 
\State Set $S_z(0)=S_0\coloneqq \beta I$ and $i=0$
\Repeat
\State Measure $x_t$
\State Compute $u_t = K_zx_t + k_z + e_t$ and apply $u_t$
\State \textcolor{red}{Receive $c_t$} and measure $x_{t+1}$
\State Compute $\hat{\theta}_z(i+1)$ using RLS (\ref{Eq:RLS})
\State Set $t=t+1$ and $i=i+1$
\Until{$\Vert \hat{\theta}_z(i) -\hat{\theta}_z(i-1) \Vert_2 < \epsilon_2 $  }
\State Find the matrix $\hat{H}$ corresponding to  $\hat{\theta}_z(i)$
\State Compute $K_{z+1} = -\hat{H}_{(uu)}^{-1}\hat{H}_{(ux)}$ 
\State Compute $k_{z+1} = -\hat{H}_{(uu)}^{-1}\hat{H}_{(u1)}$
\State Set $\hat{\theta}_{z+1}(0)=\hat{\theta}_z(i)$ and z = z+1
\Until{$\Vert K_z -K_{z-1} \Vert_F + \Vert k_z -k_{z-1} \Vert_F < \epsilon_2$}
\end{algorithmic}
\end{algorithm}

\subsection{The ADP Attacker}

We define the ADP attack model by specifying the attacker's capability, objective, and the information he has: 
\begin{enumerate}
    \item The attacker can falsify the cost signals $c_t$ and $c_t^\dag$ is received by the agent in line 8 of Algorithm \ref{Algo:ADPPolicyIter}.
    \item The attacker's objective is to trick the agent into learning the `nefarious' policy $(K^\dag,k^\dag)$.
    \item The attacker knows the LQG system and receives the same information as the agent during the learning process.
\end{enumerate}

The problem for the attacker is how to craft such an attack so that the agent will finally learn the `nefarious' policy and the falsified cost $c_t^\dag$ deviate insignificantly from the original cost signals $c_t$. Hence, the attacker can first solve the cost parameters falsification problem (\ref{Eq:CostParametersOptimization}). Let $(D^\dag, E^\dag,d^\dag)$ be the solution of problem (\ref{Eq:CostParametersOptimization}). The falsified cost signal $c_t^\dag$ then can be crafted using $c_t^\dag = x_t D^\dag x_t + d^\dag x_t + r + u_t'E u_t$ for $t=0,1,2,\cdots$.  The agent receives $c_t^\dag$ in Algorithm \ref{Algo:ADPPolicyIter} and other information received remain correct. Following the same arguments as in Theorem 1 of \cite{bradtke1994adaptive}, we know that if problem (\ref{Algo:ADPPolicyIter}) is feasible, the agent will eventually be tricked into learning the `nefarious' policy $(K^\dag,k^\dag)$. In Section \ref{Sec:NumericalStudy}, we will demonstrate the effectiveness of the ADP attack model and the vulnerabilities of the ADP based LQG systems.

\section{Numerical Studies}\label{Sec:NumericalStudy}

We now use an LQG system to demonstrate the bounds we obtain in Section \ref{Sec:LQGManiCostPara} as well as the effectiveness of the two attack models against the Batch RL learner in Section \ref{Sec:BarchRLLearner} and the ADP learner in Section \ref{Sec:DecevADPLearner}. Consider a linear system with $6$-dimensional state consisting of its $3$-dimensional (3D) position and $3$-dimensional (3D) velocity:
$$
A = \begin{bmatrix}
I_3 & 0.1I_3 \\
\mathbf{0} & 0.95I_3\\
\end{bmatrix}\in\mathbb{R}^{6\time 6},\ \ \ \ B = \begin{bmatrix}
\mathbf{0}\\
0.1 I_3\\
\end{bmatrix}\in\mathbb{R}^{6\times 3}.
$$
Suppose $C=I_n$ and $w_t\sim \mathcal{N}(0,0.01I_n)$. Let $x_t = [\chi_t,\eta_t,\zeta_t,v^\chi_t,v^\eta_t,v^\zeta_t]'$. Here, $(\chi_t,\eta_t,\zeta_t)$ is the position of the vehicle in the 3D space at time $t$ and $(v^\chi_t,v^\eta_t,v^\zeta_t)$ is the velocity at time $t$. The vehicle starts from the initial position $(\chi_0,\eta_0,\zeta_0) = (1,1,0.5)$ with velocity $(v^\chi_0,v^\eta_0,v^\zeta_0) = (-1,-0.5,-1)$. The LQG aims to stabilize the vehicle to the origin. The true cost parameters are $D=I_n$, $E= 0.5 I_m$, $d=0$, and $r=0$ and the cost signals the agent receives is $c_t = x_t'Dx_t + d'x_t + r + u_t'Eu_t$. The discount factor $\gamma$ is set to be $\gamma = 0.9$. The optimal policy, which stabilizes the vehicle to the origin, can be computed according to (\ref{Eq:ControllerParameters})-(\ref{Eq:LinearEquation}):
\begin{equation}\label{Eq:NumOptimalPolicy}
\scriptsize
\begin{aligned}
&K^* =\begin{bmatrix}
 -0.5316 &  -0.0000 & -0.0000 & -0.9700 & -0.0000 & -0.0000\\
 -0.0000 & -0.5316 &  -0.0000 & -0.0000 &  -0.9700 & -0.0000\\      
 -0.0000 & -0.0000 & -0.5315 &  -0.0000 &       -0.0000 & -0.9700\\
\end{bmatrix},\\
&k^* = \begin{bmatrix}
-0.0000\\
-0.0000\\
-0.0000\\
\end{bmatrix}.
\end{aligned}
\end{equation}

Suppose the attacker aims to trick the agent into learning the following policy:
\begin{equation}\label{Eq:NumNefariousPolicy}
\scriptsize
\begin{aligned}
&K^\dag = \begin{bmatrix}
-0.5316 &   0   & 0  & -0.9700 &   0 & 0\\
0   &   -0.5316  & 0  & 0. & -0.9700  & 0\\
0   &  0 & -0.5316  & 0  &  0 & -0.9700
\end{bmatrix}\\
&k^\dag = \begin{bmatrix}
0.5316\\
0\\
-0.5316\\
\end{bmatrix},
\end{aligned}
\end{equation}
which steer the vehicle to a dangerous position at $ (\chi,\eta,\zeta) = (1,0,-1)$
and stay there with zero velocity.

To see whether he can successfully mislead the agent into learning the `nefarious' policy $(K^\dag,k^\dag)$ by falsifying the cost parameters, the attacker can check the conditions in Theorem \ref{Theo: AttackFeasibilityCondition}. Indeed, the conditions 1), 2), and 3) are held when choosing $E=I_m$ meaning the cost parameter falsification problem (\ref{Eq:CostParametersOptimization}) is feasible. Solving problem $(\ref{Eq:CostParametersOptimization})$ using CVXPY \cite{diamond2016cvxpy} yields (we write numerical values that are less than $1.0\mathrm{e}{-10}$ as `0' due to space limitation.)
\begin{equation}\label{Eq:NumFalsifiedCostPrameters}
\scriptsize
\begin{aligned}
D^\dag &=\begin{bmatrix}
0.7163 & 0 & 0.2837 & -0.1218 & 0 & 0.1218\\
0 &  1.000  & 0 & 0 & 0 & 0\\
0.2837 &  0 &  0.7163 & 0.1218 & 0 & -0.1218\\
-0.1218 & 0 &  0.1218 &  0.5687 &0 &  0.4313\\
0 & 0 & 0 & 0 & 1.000 & -0\\
0.1218 & 0 & -0.1219 &  0.4313 & 0 &  0.5687
\end{bmatrix},\\
d^\dag &= \begin{bmatrix}
-0.1448 & 0 & 0.1448 & -1.6084 & 0 & 0.1608
\end{bmatrix}',
\\
E^\dag &= \begin{bmatrix}
0.2904 & 0 & 0.2096\\
0 & 0.5000 & 0\\
2.095& 0 & 0.2904\\
\end{bmatrix}.
\end{aligned}
\end{equation}
The optimal value $\Vert D^\dag - D \Vert_F + \Vert E^\dag - E \Vert_F+ \Vert d^\dag - d \Vert_F$ of the optimization problem (\ref{Eq:CostParametersOptimization}) is $1.8137$. The policy $(\hat{K}^\dag, \hat{k}^\dag)$ computed using the falsified cost parameters $D^\dag,E^\dag,d^\dag$ aligns well with the attacker's target policy $(K^\dag, k^\dag)$ stated in (\ref{Eq:NumNefariousPolicy}):
$$
\scriptsize
\begin{aligned}
&\hat{K}^\dag =\begin{bmatrix}
 -0.5315 &  -0.0000 & -0.0000 & -0.9699 & -0.0000 & -0.0000\\
 -0.0000 & -0.5315 &  -0.0000 & -0.0000 &  -0.9699 & -0.0000\\ 
 -0.0000 & -0.0000 & -0.5315 &  -0.0000 &  -0.0000 & -0.9699\\
\end{bmatrix},\\
&\hat{k}^\dag = \begin{bmatrix}
0.5316\\
0.0000\\
-0.5316\\
\end{bmatrix}.
\end{aligned}
$$

\subsection{Attacking Batch RL Learner}

The original training data set $\mathcal{D} = \{(x_t,u_t,c_t,x_{t+1}),t=0,1,\cdots,T-1\}$ is generated by running the LQG system with uniformly distributed random controls and receiving the accurate cost signals $c_t,t=0,1,\cdots,T-1$. Here, $T=400$ is the number of time steps (i.e., the number of data tuple collected).

Based on the dataset $\mathcal{D}$, the Batch RL learner estimates the system parameters and the cost parameters using (\ref{Eq:CostParametersEstimation}) and (\ref{Eq:CostParametersEstimation}). With the estimates, the computed optimal policy under the clean data is
$$
\scriptsize
\begin{aligned}
\hat{K}_{\textrm{batch}}^* &= \begin{bmatrix}
-0.5613 & -0.0148 & 0.0234 & -0.9878 & 0.0188 & 0.0063\\
0.0268 & -0.5232 & -0.0045 & 0.0265 & -0.9337 & -0.0189\\
-0.0047 & -0.0145 & -0.5465 & 0.0171 & -0.0111 & -0.9630
\end{bmatrix},\\
\hat{k}_{\mathrm{batch}}^* &=\begin{bmatrix}
-1.5601\mathrm{e}{-06} & -3.9701\mathrm{e}{-07} & 1.6453\mathrm{e}{-06}
\end{bmatrix}'.
\end{aligned}
$$
We can see that the batch RL learner learns a decent policy $(\hat{K}_{\mathrm{batch}}^*,\hat{k}_{\mathrm{batch}}^*)$ that only differs slightly from the optimal policy (\ref{Eq:NumOptimalPolicy}).

\begin{figure}[ht]
  \begin{subfigure}[b]{1\columnwidth}
    \includegraphics[width=\columnwidth]{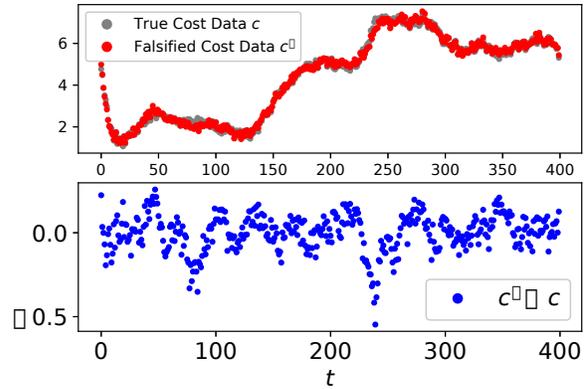}
    \subcaption{True cost data, falsified cost data, and their differences.}
    \label{fig:AttackingBatchRL-a}
  \end{subfigure}
  \vfill 
  \begin{subfigure}[b]{1\columnwidth}
    \includegraphics[width=\columnwidth]{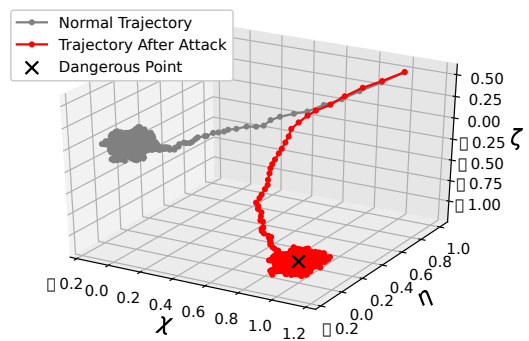}
    \subcaption{Trajectories generated from the policy learned from the clean data and the policy learned from the falsified data respectively.}
    \label{fig:AttackingBatchRL-b}
  \end{subfigure}
  \caption{Demonstration of attacks against LQG agent using batch RL.}
  \label{Fig:AttackingBatchRL}
\end{figure}

The attacker aims to trick the agent into learning the `nefarious' policy (\ref{Eq:NumNefariousPolicy}) by falsifying the cost data $c_t$ into $c_t^\dag$. First, the attacker estimates the system parameters using the clean dataset $\mathcal{D}$. Then, the attacker solves the optimization problem (\ref{Eq:ObjectiveBatchRL}) to generate the falsified cost data $c_t^\dag,t=0,1,\cdots,T-1$. The LQG agent learns from the falsified data set $\mathcal{D}^\dag = \{(x_t,u_t,c^\dag_t,x_{t+1}),t=0,1,2,\cdots,T-1\}$ and is misled into learning the policy
$$
\scriptsize
\begin{aligned}
\hat{K}_{\textrm{batch}}^\dag &= \begin{bmatrix}
-0.5613 & -0.0148 & 0.0234 & -0.9878 & 0.0188 & 0.0063\\
0.0268 & -0.5232 & -0.0045 & 0.0265 & -0.9337 & -0.0189\\
-0.0047 & -0.0145 & -0.5465 & 0.0171 & -0.0111 & -0.9630
\end{bmatrix},\\
\hat{k}_{\textrm{batch}}^\dag &=\begin{bmatrix}
0.5812 & -0.0382 & -0.5310
\end{bmatrix}',
\end{aligned}
$$
which is close to the attacker target policy (\ref{Eq:NumNefariousPolicy}).

The results of the attack model are demonstrated in Fig. \ref{Fig:AttackingBatchRL}. Fig. \ref{fig:AttackingBatchRL-a} presents the true cost data, the falsified cost data, and how much falsification is imposed. We can see that the attacker only applied small falsification on each cost signal (i.e., the cost data and the falsified cost data stays close). Fig. \ref{fig:AttackingBatchRL-b} shows the trajectories of the vehicle under the policy $(\hat{K}_{\mathrm{batch}}^*,\hat{k}_{\mathrm{batch}}^*)$ learned from the true data $\mathcal{D}$ and the policy $(\hat{K}_{\textrm{batch}}^\dag, \hat{k}_{\textrm{batch}}^\dag)$ learned from the falsified data $\mathcal{D}^\dag$. We can see that the attack successfully tricks the LQG agent into learning a policy that steers the vehicle into a dangerous position. The attack's success is indeed built on a very small falsification of the cost data. To see this, we have $\Vert \mathbf{c}^\dag - \mathbf{c} \Vert_2 / \Vert \mathbf{c} \Vert_2= 2.2107/96.2731=2.296\%$. The results show the attack model's effectiveness and exhibit the LQG agent's vulnerabilities using batch RL.

\begin{figure}[ht]
  \begin{subfigure}[b]{1\columnwidth}
    \includegraphics[width=\columnwidth]{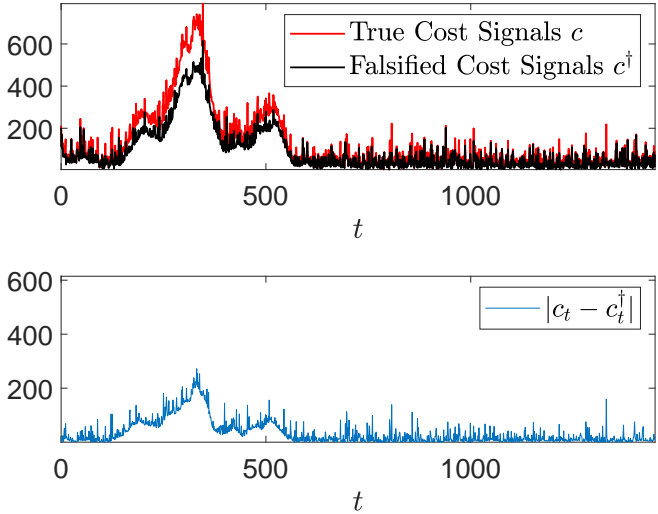}
    \subcaption{True cost signals, falsified cost signals, and their differences.}
    \label{fig:AttackingADP-a}
  \end{subfigure}
  \newline\vfill 
  \begin{subfigure}[b]{1\columnwidth}
    \centering
    \includegraphics[width=0.8\columnwidth]{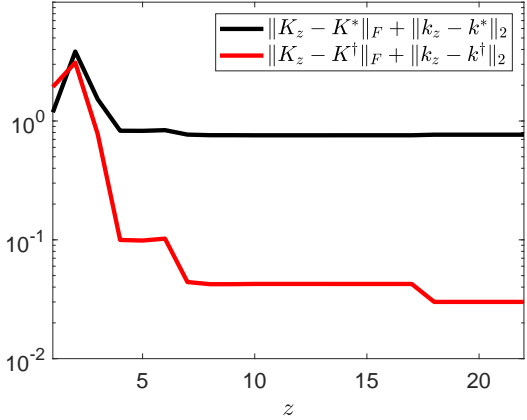}
    \subcaption{The policy updates $K_z,z=0,1,2,\cdots,22$ during the ADP learning process under the falsified cost signals and the policy sequences converge to the attacker's target policy $(K^\dag,k^\dag)$ rather than the optimal policy $(K^*,k^*)$.}
    \label{fig:AttackingADP-b}
  \end{subfigure}
  \caption{Demonstration of attacks against LQG agent using ADP.}
  \label{Fig:AttackingADP}
\end{figure}

\subsection{Attacking ADP Learner}

The LQG agent implements Algorithm \ref{Algo:ADPPolicyIter} to learn the optimal policy to control the vehicle by interacting with the environment. Set the tolerance $\epsilon_1 = \epsilon_2 = 1\mathrm{e}{-5}$. Set $\beta =10$. Let the initial control policy be 
$$
\scriptsize
\begin{aligned}
K_0 &= \begin{bmatrix}
0.0300     &    0    &     0  & -0.1000       &  0     &    0\\
0  &  0.5600   &      0    &     0   & 1.0000    &     0\\
0   &      0  &  0.5000    &     0      &   0  &  1.0000\\
\end{bmatrix},\\
k_0 &= \begin{bmatrix}
0 & 0 & 0
\end{bmatrix}'.
\end{aligned}
$$
The attacker aims to trick the ADP learner into learning the `nefarious' policy (\ref{Eq:NumNefariousPolicy}) by falsifying the cost signals received by the LQG agent. As we discussed in Section \ref{Sec:DecevADPLearner}, the attacker can use (\ref{Eq:NumFalsifiedCostPrameters}) to falsify cost signals into $c_t^\dag = x_t' D^\dag x_t + {d^\dag}'x_t + u_t'E^\dag u_t +r$. The numerical computation indicates that 
$$
| c^\dag_t - c_t | \leq 1.088\Vert x_t\Vert_2^2 + 0.419\Vert u_t\Vert_2^2 + 0.3060\Vert x_t\Vert_2.
$$

If the LQG agent receives the cost signals, the ADP algorithms finally lead the LQG agent to the policy
$$
\scriptsize
\begin{aligned}
\hat{K}_{\textrm{ADP}}^* &=
\begin{bmatrix}
-0.5280 &  -0.0058 &   0.0064  & -0.9673  &  0.0001  &  0.0035\\
-0.0021 &  -0.5409  & -0.0201  &  0.0015  & -0.9742 &  -0.0166\\
0.0065  & -0.0340 &  -0.5231  &  0.0033  & -0.0177  & -0.9632
\end{bmatrix},\\
\hat{k}_{\textrm{ADP}}^* &=\begin{bmatrix}
-0.0007 & 0.0054& -0.0093
\end{bmatrix}',
\end{aligned}
$$
which only differs slightly from the optimal policy (\ref{Eq:NumOptimalPolicy}).

If the LQG agent receives the falsified cost signals, the agent will be tricked into learning the policy
$$
\scriptsize
\begin{aligned}
\hat{K}_{\textrm{ADP}}^\dag &=
\begin{bmatrix}
 -0.5253 &   0.0473   & 0.0284 &  -0.9685  &  0.0145  &  0.0242\\
 -0.0009  & -0.5274 &  -0.0050  & -0.0028  & -0.9696 &  -0.0055\\
 -0.0093 &  -0.0675 &  -0.5754. &   -0.0054  & -0.0234 &  -1.0034
\end{bmatrix},\\
\hat{k}_{\textrm{ADP}}^\dag &=\begin{bmatrix}
0.5229 & 0.0076 & -0.5144
\end{bmatrix}',
\end{aligned}
$$
which is close the the attacker's target policy (\ref{Eq:NumNefariousPolicy}).

The results of the attack model on ADP are demonstrated in Fig. \ref{Fig:AttackingADP}. In Fig. \ref{fig:AttackingADP-a}, we present the true cost signals, the falsified cost signals, and the absolute value of their differences over time. As we can see, the falsified cost signals stay close to the true cost signals. The falsification grows as $\Vert x_t\Vert$ and $\Vert u_t \Vert$ increases because the falsified cost signals is generated using $c^\dag=x_t'D^\dag x_t + u_t'E^\dag u_t +r + d^\dag x_t$. Fig. \ref{fig:AttackingADP-b} shows the how the policy $(K_z,k_z)$ (in line 13-14 Algorithm \ref{Algo:ADPPolicyIter}) iterates during the learning process. The attacker gradually misleads the LQG agent into learning the `nefarious' policy $(K^\dag,k^\dag)$. The results show the attack model's effectiveness and exhibit the LQG agent's vulnerabilities using ADP approaches.

\section{Conclusion}
In this work, we have studied the vulnerabilities of the RL-enabled LQG control systems under cost signal falsification. We have shown that a small falsification of the cost parameters will only lead to a bounded change in the optimal policy. The bound is linear on the amount of falsification the attacker can apply to the cost parameters. This result shows a certain degree of robustness of the optimal policy to small unintended changes in the cost parameters. We have proposed an attack model where the attacker conducts intentional falsification on the cost parameters to mislead the agent into learning a `nefarious' policy. We have formulated the attack's problem as an optimization problem, which is proved to be convex, and developed necessary and sufficient conditions to check the feasibility of the attacker's problem.

Based on the attack model on cost parameters, we discussed two attack models on the batch RL and the ADP learners. Numerical results have shown that with only $2.296\%$ of falsification on the cost data, the attacker can achieve his goal --- misleading the batch RL learner into learning the 'nefarious' policy that leads the vehicle to a dangerous position. An ADP learner updates the policy iteratively and learns in an online manner. The results have demonstrated that the attacker can gradually trick the learner into learning the `nefarious' policy, even if the falsified cost signals stay close to true cost signals.

The paper has exhibited the effectiveness of these attack models and revealed vulnerabilities of the RL-enabled control systems. The authors hope this paper can bring more attention from the community to the potential security threats faced by RL-enabled control systems. Future works can focus on more effective attack models on cost signals, state measurements, and control commands. With a better understanding of these attack models, we can take further steps to develop reliable and effective detection and defensive mechanisms.





\section*{APPENDIX}

\subsection{Lemmas}

\begin{lemma}\label{Lemma:NormInquality1}
Given any two positive semi-definite matrices of the same dimension $X$ and $Z$,
$$
\Vert X(I+SX)^{-1}\Vert \leq  \Vert X \Vert. 
$$
\end{lemma}
\begin{proof}
Assume $X>0$. Note that $(X^{-1} + S)X = (I + SX)$. Hence, the inverse of them are equal, i.e., $X^{-1}(X^{-1}+S)^{-1} = (I+SX)^{-1}$. Then, we obtain $X(I+SX)^{-1}=(X^{-1}+S)\leq X$. By a continuity argument, the inequality also holds when $X$ is only positive semi-definite.
\end{proof}

\begin{lemma}\label{Lemma:BoundFunctionNorm}
Let $\mathcal{L}(\mathbb{R}^{n\times n},\mathbb{R}^{n\times n})$ be the space of linear operators $\mathbb{R}^{n\times n}\rightarrow \mathbb{R}^{n\times n}$ with the following operator norm
$$
\Vert \mathcal{T}\Vert_\mathcal{L} = \max \{\Vert \mathcal{\mathcal{T}}(X)\Vert:\Vert X\Vert =1 \}, \ \ \ \textrm{for } \mathcal{T}\in \mathcal{L}(\mathbb{R}^{n\times n},\mathbb{R}^{n\times n}). 
$$
Define $\mathcal{T}(X) = X - \gamma A_c' X A_c$ with $0<\gamma \leq 1$ and $A_c$ is Schur stable. There exists a finite $\tau(A_c)$ such that $\Vert \mathcal{T}^{-1} \Vert_\mathcal{L}\leq \frac{\tau(A_c)^2}{1-\gamma \rho(A_c)^2}$
\end{lemma}
\begin{proof}
Note that $A_c$ is Schur stable. There exists $L$ such that $M = X - \gamma A_c' XA_c$, which is a discrete-time Lyapunov equation given $M$ and $A_c$. It is easy to verify that $X= \sum_{k=0}^\infty \gamma^k (A_c')^k M (A_c)^k$ by pluging it back into the Lyapunov equation. Hence, we have
$$
\begin{aligned}
\left\Vert \mathcal{T}^{-1}(M) \right\Vert &= \left\Vert \sum_{k=0}^\infty \gamma^k(A_c')^kM(A_c)^k \right\Vert \\
&\leq \sum_{k=0}^\infty \gamma^k \Vert A_c^k \Vert^2 \Vert M \Vert. 
\end{aligned}
$$
By Gelfand's formula, there exists a finite $\tau(M)$ such that $\Vert A_c^k \Vert \leq \tau(A_c) \rho(A_c)^k$ for every $k\geq 0$. Here, $\rho(A_c)$ is the spectral radius of $A_c$ and $\tau(A_c)$ is defined as $\tau(A_c)\coloneqq \sup_{k\in\mathbb{N}} \{\Vert A_c^k \Vert/ \rho(A_c)^k\}$. Since $\rho(A_c)<1$, $\sum_{k=1}^\infty \gamma^k \Vert A_c\Vert^k$ is convergent and $\Vert \mathcal{T}^{-1}\Vert_\mathcal{L}\leq \sum_{k=0}^\infty \gamma^k \Vert A_c^k \Vert^2\leq \frac{\tau(A_c)^2}{1-\gamma \rho(A_c)^2}$.
\end{proof}

\begin{lemma}\label{Eq:NormInquality2}
Let $\Delta E = E_2-E_1$, where $E_1$ and $E_2$ are two square matrices $E_1$ and $E_2$. Assume $\Vert \Delta E\Vert<1$ and $\Vert E^{-1}_1 \Delta E  \Vert<1$. We have
$$
\Vert E_2^{-1} - E_1^{-1} \Vert \leq \frac{\Vert E_1^{-1} \Vert^2 \Vert \Delta E\Vert}{1-\Vert E_1^{-1}\Vert}.
$$
\end{lemma}
\begin{proof}
A straight calculation gives $E_1^{-1} - E_2^{-1} = E_1^{-1}\Delta E E_2^{-1}.$
It follows immediately that 
\begin{equation}\label{Eq:ImmediateStep2}
\Vert E_1^{-1} - E_2^{-1} \Vert = \Vert E_1^{-1}\Delta E E^{-1}_2\Vert\leq  \Vert E_1^{-1}\Delta E \Vert \Vert E^{-1}_2\Vert.
\end{equation}
Applying the triangle inequality, we obtain $\Vert E_2^{-1}  \Vert \leq \Vert E_1^{-1}\Vert + \Vert E_1^{-1} \Delta E\Vert \Vert E_2^{-1}\Vert\leq \frac{\Vert E_1^{-1}\Vert}{1-\Vert E_1^{-1} \Delta E\Vert}.$ Plugging it back into (\ref{Eq:ImmediateStep2}) yields
$$
\Vert E_2^{-1} - E_1^{-1} \Vert \leq  \frac{\Vert E_1^{-1}\Vert^2 \Vert \Delta E\Vert}{1-\Vert E_1^{-1} \Delta E\Vert}.
$$
Note that given $\Vert \Delta E\Vert <1$, $$\Vert E_1^{-1} \Delta E \Vert\leq \Vert   E_1^{-1}\Vert \Vert \Delta E\Vert\leq \Vert E_1^{-1}\Vert.
$$
Hence, 
$$
\Vert E_2^{-1} - E_1^{-1} \Vert \leq  \frac{\Vert E_1^{-1}\Vert^2}{1-\Vert E_1^{-1} \Vert}   \Vert \Delta E\Vert.
$$
\end{proof}

\begin{lemma}\label{Lemma:InvertibleMatrix}
Given a real-valued Schur stable matrix $A_c$ and a scalar $0<\gamma\leq 1$, $I-\gamma A_c$ is invertible.
\end{lemma}
\begin{proof}
Suppose that $\lambda$ is an eigenvalue of $ A_c$ and $v$ is its associated vector. Hence, we have
$$
A_cv=\lambda v.
$$
Therefore, we have
$$
(I-\gamma A_c)v=  v-\gamma \lambda v = (1-\gamma \lambda) v. 
$$
Hence, for every eigenvalue $\lambda$ of matrix $A_c$, $1-\gamma \lambda$ is an eigenvalue of the matrix $I-\gamma A_c$. Since $A_c$ is Schur stable, $\vert \lambda \vert<1$ for every eigenvalue $\lambda$ of matrix $A_c$. Hence, the real part of $1-\gamma \lambda$ is non-zero for every eigenvalue $\lambda$ of matrix $A_c$. Thus, $I-\gamma A_c$ is invertible.
\end{proof}

\begin{lemma}\label{Lemma:SolutionBoundConvexFun}
Suppose $f_1,f_2\in\mathbb{R}^n \rightarrow \mathbb{R}$ are twice differentiable and $\alpha$-strongly convex, i.e., $\nabla^2 f_1\leq \alpha I$ and $\nabla^2f_2\leq \alpha I$. Let $u_1 = \argmin_u f_1(u)$ and $u_2 =\argmin_u f_2(u)$. Suppose $\Vert f_1(u_2)\Vert \leq \epsilon$, then $\Vert u_1-u_2 \Vert \leq \frac{\epsilon}{\alpha}$.
\end{lemma}
\begin{proof}
An application of Taylor expansion yields
$$
\nabla f_1(u_2) = \nabla f_1(u_1) + \nabla^2 f_1(\bar{u})(u_2-u_1),
$$
with $\bar{u} = \beta u_1 + (1-\beta)u_2$ for some $0\leq\beta\leq 1$. Since $u_1$ is a stationary point of $f_1$, i.e., $\nabla f_1(u_1)=0$, we have
$$
\Vert \alpha I (u_1 -u_2)\Vert \leq \Vert \nabla^2 f_1(\bar{u})(u_1-u_2)\Vert = \Vert \nabla f_1(u_2)\Vert \leq \epsilon,
$$
which gives $\Vert u_1 - u_2\Vert\leq \frac{\epsilon}{\alpha}$.
\end{proof}

\begin{lemma}\label{Lemma:BoundonController}
Define the functions $f_i(x,u)= \frac{1}{2} u'E_i u +\gamma \frac{1}{2} (Ax+Bu)' P_i(Ax+Bu)$ for $i=1,2$. Suppose $E_1$ and $E_2$ are positive definite. Suppose $\Vert E_1-E_2\Vert \leq \epsilon$ and $\Vert P_1-P_2 \Vert\leq f(\epsilon)$ with $\epsilon\leq \lambda_{\min}(E_1)/2$, where $\lambda_{\min}(\cdot)$ is the smallest eigenvalue of a matrix. Let $K_i x = \argmin_u f_i(x,u)$. Then, we have

\begin{equation}\label{Eq:BoundonK}
\begin{aligned}
\Vert K_1-K_2 \Vert \leq& \frac{2\gamma}{\lambda_{\min}(E_1)}\max\{\Vert A\Vert,\Vert B\Vert \}^2(\Vert K_1\Vert +1) f(\epsilon) \\
&+\frac{2\gamma}{\lambda_{\min}(E_1)} \Vert K_1\Vert \epsilon.
\end{aligned}
\end{equation}
\end{lemma}
\begin{proof}
According to the definition of strong convexity, $f_1,f_2$ are $\alpha$-strongly convex in $u$ with $$
\alpha =\min \{ \lambda_{\min\ }(E_1),\ \lambda_{\min\ }(E_2) \}.$$
We can calculate the gradient of $f_i(x,u)$ as 
$$
\nabla f_i(x,u)= (\gamma B'P_iB+E_i)u + \gamma B'P_iAx.
$$
It is obvious that 
$$
\Vert \gamma B'P_1 B +E_1 - \gamma B'P_2 B -E_2\Vert \leq \gamma \Vert B\Vert^2\Vert P_1 - P_2\Vert + \Vert E_1 -E_2\Vert,
$$
and
$$
\Vert \gamma B'P_1A - \gamma B'P_2 A\Vert \leq \gamma \Vert B\Vert \Vert A\Vert \Vert P_1 - P_2\Vert.
$$
Hence, for any $x$ such that $\Vert x\Vert \leq 1$, we have
$$
\begin{aligned}
&\Vert \nabla f_1(x,u) -\nabla f_2(x,u)\Vert\\
\leq &\left[ \gamma \Vert B\Vert^2\Vert P_1 - P_2\Vert + \Vert E_1 -E_2\Vert\right]\Vert u\Vert\\
&\ + \left[\gamma \Vert B \Vert \Vert A\Vert \Vert P_1-P_2\Vert \right]\\
\leq & \gamma \max\{ \Vert A\Vert,\Vert B\Vert \}^2 \Vert P_1 -P_2\Vert (\Vert u\Vert +1) + \gamma \Vert E_1- E_2\Vert \Vert u\Vert
\end{aligned}
$$
From Lemma \ref{Lemma:SolutionBoundConvexFun}, we know
$$
\alpha \Vert K_1 x - K_2 x\Vert =\alpha \Vert u_1 - u_2\Vert \leq \Vert \nabla f_2(u_1)\Vert. 
$$
Note that $\Vert x\Vert \leq 1$ and $\Vert u_1\Vert \leq \Vert K_1 x\Vert\leq \Vert K_1\Vert$. We have
\begin{equation}\label{Eq:BoundOnIntermediateK}
\begin{aligned}
\Vert K_1 -K_2\Vert \leq & \frac{\gamma}{\alpha}\max\{\Vert A\Vert,\Vert B\Vert \}^2(\Vert K_1\Vert +1)\Vert P_1 - P_2\Vert \\
&+\frac{\gamma}{\alpha} \Vert K_1\Vert \epsilon.
\end{aligned}
\end{equation}
Note that $\epsilon \leq \lambda_{\min}(E_1)/2$. By Weyl's inequality $ \lambda_{\min}(E_1)/2 \leq  \alpha$, we can show (\ref{Eq:BoundonK}) from (\ref{Eq:BoundOnIntermediateK}).
\end{proof}

\begin{lemma}\label{Lemma:PerturbationLinearSystem}
Let $M\in\mathbb{R}^{n\times n}$ be a non-singular matrix, and $\Delta M \in\mathbb{R}^{n\times n}$. Assume $\Vert \Delta M\Vert \leq \frac{1}{2\Vert M^{-1}\Vert}$. Let $h\in\mathbb{R}^n$ and $h+
\Delta h \in\mathbb{R}^{n}$ be the solution of 
$$
Mh=d,\ \ \ (M+\Delta M)(h+\Delta h) = d+\Delta d,
$$
for some $d\in\mathbb{R}^n$ and $\Delta d \in\mathbb{R}^n$. Then, $M+\Delta M$ is non-singular and
$$
\Vert \Delta h\Vert \leq 2\Vert M^{-1} \Vert 
\left(\Vert \Delta d\Vert +\Vert \Delta M\Vert \Vert h \Vert   \right).
$$
\end{lemma}
\begin{proof}
It is easy to see that 
$$
M+\Delta M = M(I+M^{-1}\Delta M).
$$
Since $M$ is non-singular, if $I+M^{-1}\Delta M$ is non-singular, $M+\Delta M$ is non-singular. Indeed, for every non-zero $x\in\mathbb{R}^n$,
$$
\begin{aligned}
\Vert (I + M^{-1}\Delta M)x \Vert &\geq \Vert x\Vert - \Vert M^{-1}\Delta M\Vert \Vert x\Vert\\
&=(1-\Vert M^{-1}\Delta M\Vert) \Vert x\Vert\\
&\geq 0,
\end{aligned}
$$
which shows the non-singularity of $M+\Delta M$.
Now, we can solve for $\Delta h$:
$$
\begin{aligned}
\Delta h &= (M+\Delta M)^{-1}(\Delta d - \Delta M h)\\
&= (I + M^{-1}\Delta M)^{-1} M^{-1} (\Delta d - \Delta M h)\\
\end{aligned}
$$
Hence, we have
$$
\begin{aligned}
\Vert \Delta h\Vert &\leq \Vert (I+M^{-1}\Delta M)^{-1}\Vert \Vert M^{-1}\Vert (\Vert \Delta d\Vert  +\Vert \Delta M\Vert \Vert h\Vert )\\
&\leq \frac{\Vert M^{-1}\Vert}{1-\Vert M^{-1}\Vert \Vert \Delta M\Vert} \left(\Vert \Delta d\Vert + \Vert \Delta M\Vert \Vert h \Vert \right)\\
&\leq 2\Vert M^{-1}\Vert \left( \Vert \Delta d\Vert + \Vert \Delta M\Vert \Vert h \Vert \right).
\end{aligned}
$$
\end{proof}

\subsection{Proof of Lemma \ref{Lemma:PolicyUniquelyDecided}}\label{Proof:PolicyUniquelyDecided}
\begin{proof}
From Theorem 4 of \cite{kushner1971introduction}, we know that if $(A,B)$ is controllable and $(A,D^{1/2})$ is observable, the solution to the Riccati equation $(\ref{Eq:RiccatiEquation})$ is positive definite and unique and $K$ is stabilizing. Then from (\ref{Eq:ControllerParameters}), we know $K^*$ is uniquely decided by the system parameters $(A,B)$ and the cost parameters $(D,E)$.

Since $K$ is stabilizing, from Lemma \ref{Lemma:InvertibleMatrix}, $I_n - \gamma(A+BK)$ is invertible. Hence, the solution $h^*$ of  (\ref{Eq:LinearEquation}) is unique and depend only on the control parameters $(A,B)$ and $K^*$. Then, (\ref{Eq:ControllerParameters}) shows that $k^*$ is uniquely decided by $(A,B)$ and $(D,E,d)$.
\end{proof}






\bibliography{references}

\bibliographystyle{IEEEtran}

\end{document}